\documentclass[a4paper,12pt]{article}
\usepackage{amssymb,amsmath,amsthm}
\usepackage{texdraw}
\usepackage{a4wide}
\usepackage{subcaption}
\usepackage{pgf,tikz}
\usetikzlibrary{arrows}
\usepackage{graphics}
\usepackage{graphicx}
\usepackage{epsfig}

\newtheorem{theorem}{Theorem}[section]

\newtheorem{proposition}[theorem]{Proposition}
\newtheorem{lemma}[theorem]{Lemma}
\newtheorem{corollary}[theorem]{Corollary}

\theoremstyle{definition}

\newtheorem{example}[theorem]{Example}
\newtheorem{remark}[theorem]{Remark}

\begin{document}


\title{ A new class of integrable Lotka--Volterra systems}
\author{H. Christodoulidi$^{1,2}$, A.N.W. Hone$^{2,3}$  and T.E. Kouloukas$^{2}$\\  
\ \\ 
$^{1}$ Research Center for Astronomy and Applied Mathematics  
\\ Academy of Athens, Athens 11527, Greece\\
$^{2}$School of Mathematics, Statistics and Actuarial Science, \\
  University of Kent, Canterbury CT2 7NF, UK\\
$^{3}$ School of Mathematics and Statistics, \\
University of New South Wales, Sydney NSW 2052, Australia}

\maketitle

\begin{abstract}
 A parameter-dependent class of Hamiltonian (generalized) Lotka--Volterra systems is considered. We prove that this class contains Liouville integrable as well as superintegrable cases according to particular choices of the parameters.  We determine  sufficient conditions which
result in integrable behavior, while we numerically explore the complementary cases, 
where these analytically derived conditions are not satisfied.  
\end{abstract}

\section{Introduction}

 {The} Lotka--Volterra system  {was} introduced independently 
by Lotka \cite{lotka} and Volterra  \cite{volterra} as a predator-prey model. 
Since then, many generalizations have been considered with applications to several scientific disciplines.
 {These systems in general display rich dynamical behavior that varies according to the 
parameters that define each one of them.} For example, there are Hamiltonian and 
non-Hamiltonian Lotka--Volterra systems, as well as integrable, 
non-integrable and chaotic ones. 
From the point of view of integrability, various kinds 
of generalized Lotka--Volterra systems have been extensively studied  in the literature, e.g. 
 \cite{balle,bountis,chara,DEKV,hern,itoh87,itoh09,ragni,suris99,veselov}. 
 A  numerical study of a $4$--dimensional non-integrable 
Lotka--Volterra system can be found in \cite{vano}.

In this paper, we study a parametric family of (generalized) Lotka--Volterra systems of the form
\begin{equation}\label{system}
  \dot x_i=x_i\left(\sum_{j>i}a_jx_{j}-\sum_{j<i}a_jx_{j}+r_i\right)\;, \ \ a_i,r_i \in \mathbb{R}. 
\end{equation}
This family includes some particular interesting cases. The case of $r_i=0$ and $a_i=1$   {came up in the study of a class of multi-sums of products in [13] which is related to integrals of periodic reductions of discrete integrable systems.} It  
can be considered as a finite dimensional reduction 
of a  {Bogoyavlenskij} lattice \cite{bog87,bog} with fixed boundary conditions. The integrability of this case 
and its corresponding Kahan discretization has been studied in detail in \cite{KKQTV}. 
In \cite{KQV}, the Liouville integrability and superintegrability of the more general cases,  with $r_i=0$ and arbitrary $a_i \in \mathbb{R}$, was proved and explicit solutions were given for the corresponding continuous and discrete systems.  
 {Motivated by these results, we aim to study} the integrable and dynamical aspects of \eqref{system}, with arbitrary parameters $a_i$ and $r_i$ in $\mathbb{R}$.  

 {As is shown in Section 3, all the even-dimensional cases of  \eqref{system} are Hamiltonian with respect to a log-canonical Poisson bracket and this also applies to odd dimensions under some extra conditions on the parameters $r_i$.} A first approach to trace integrable cases is the following. We consider the integrals of the $r_i=0$ case as they appear in \cite{KQV}, and we  {demand them} to be in involution with the Hamiltonian function of \eqref{system}. This restriction leads to a system for the parameters $a_i$ and $r_i$. Solutions of this system provide necessary and sufficient conditions 
which ensure the pairwise involutivity of all the integrals (including the Hamiltonian). 
This procedure provides  {several} Liouville integrable cases. By considering a permutation symmetry of the system more integrable cases appear as well as superintegrable cases according to particular choices of the parameters. These results appear  {in} Sections $4$--$5$.

In Section $6$, we numerically 
explore the behavior of \eqref{system} with $n=4$ for the cases where integrability is not proven 
by the analytical arguments of the previous sections. To this end we perform a series of numerical simulations
for various different parameters which determine the system \eqref{system}.
Integrability or non-integrability is manifested by the  Poincar{\'e} surfaces of 
section as well as the evolution of the largest Lyapunov  
exponent for various initial conditions at gradually increasing energies. We have strong indications that 
more integrable cases exist, however, we find non-integrable cases as well. 
Notable non-integrable examples are found for the $4$-dimensional Lotka--Volterra  
system \eqref{system} with bounded trajectories in phase space, whose orbits demonstrate a particularly rich complexity.

\section{A class of Lotka--Volterra systems}
{\it Generalized Lotka--Volterra} or just {\it Lotka--Volterra systems} are systems of the form 
\begin{equation} \label{system1}
  \dot{x_i}=x_i \left(\sum_{j=1}^n A_{ij} x_j+r_i \right)\;, \qquad i=1,\dots,n\;,
\end{equation}
where $A=(A_{ij})$ is any arbitrary $n \times n$ matrix, known as the {\it community matrix} and $\mathbf{r}=(r_1,\ldots,r_n)$ is a vector in $\mathbb{R}^n$.

In this paper, we are going to study a particular class of Lotka--Volterra systems, with community matrix 
\begin{equation}\label{Amatrix}
  A=\begin{pmatrix}
  0 & a_2& a_3& \dots&a_n \\
  -a_1 & 0& a_3& \dots&a_n \\
  -a_1 & -a_2& 0& \dots&a_n \\
  \vdots & \vdots & \vdots & \ddots \\
  -a_1 &-a_2&-a_3& \dots&0
\end{pmatrix}\;, 
\end{equation}
and parameters $a_1, \ldots, a_n \in \mathbb{R}$. In this case, system \eqref{system1} can be written as \eqref{system},  
or equivalently, as 
\begin{equation} \label{system2}
  \dot{x_i}=x_i \left(\sum_{j=1}^n P_{ij} a_j x_j+r_i \right)\;,
\end{equation}
where $P$ is the antisymmetric matrix 

\begin{equation} \label{prtder}
P_{ij}= \left\{
    \begin{array}{cc}
        1-\delta_{ij}  & \mbox{for } 1 \leq i \leq j \leq n\;,  \\
          -1  & \mbox{for } 1 \leq j<i \leq n\;.
    \end{array}
    \right.
\end{equation}

The special case of \eqref{system}  
with $\mathbf{r}=0$  {was} extensively studied in \cite{KKQTV, KQV}, 
where the Liouville and superintegrability of the 
corresponding systems were proved and explicit solutions  {were} given. Here, our aim is to investigate the integrability of particular cases with $\mathbf{r}\neq 0$. In due course we mainly restrict our attention to the case that 
$n$ is even.

\section{Hamiltonian formalism} \label{secHamForm}
We consider the log-canonical Poisson structure 
\begin{equation}\label{poisson}
\{x_i,x_j\}=x_ix_j\;,\quad 1\leq i<j \leq n.
\end{equation}%
The rank of this Poisson structure, for $x_1\ldots x_n \neq0$, is $n$ for even $n$, and $n-1$ for odd $n$. In the odd case, 
 $ C:=\frac{x_1 x_3 \dots x_{n}}{x_2 x_4 \dots x_{n-1}}\;$ is a Casimir function. 

\begin{proposition} \label{prop1}
For any even $n$, $a_i, r_i \in \mathbb{R}$ and $x_i>0$, $i=1,2,\ldots , n$, the Lotka--Volterra system 
\eqref{system} is Hamiltonian with respect to the Poisson structure \eqref{poisson} and the Hamiltonian function 
$$H(\mathbf{x})=\sum_{i=1}^n (a_i x_i+k_i\log x_i),$$
where $\mathbf{x}=(x_1,\dots,x_n)$ and $\mathbf{k}=(k_1,\ldots k_n)$ defined by $\mathbf{k}=P^{-1}\mathbf{r}$. 
\end{proposition}

 {In terms of} the parameters $k_i$ , the system is written as

\begin{equation} \label{system2}
  \dot{x_i}=x_i \left(\sum_{j=1}^n P_{ij}( a_j x_j+k_j) \right)\;.
\end{equation}

For odd $n$, the matrix $P$ is not invertible. Hence, the Hamiltonian structure of Prop. \ref{prop1} does not include 
all the cases of  \eqref{system} for arbitrary $\mathbf{r} \in \mathbb{R}^n$. However, for any $n$ we can  restrict our analysis to the Hamiltonian systems 
\eqref{system2}, i.e. systems \eqref{system} with $\mathbf{r}=P\mathbf{k}$. 

By setting $u_i=\log{x_i}$, the Poisson bracket \eqref{poisson} becomes a constant one, that is 
$\{u_i,u_j\}=P_{ij}$,  and the Hamiltonian function 
$H_u=\sum_{i=1}^n (a_i e^{u_i}+k_i u_i)$. In these coordinates our system is expressed as 
\begin{equation*} 
  \dot{u_i}= \sum_{j=1}^n P_{ij} a_j e^{u_j}+r_i =\sum_{j>i}a_je^{u_j}-\sum_{j<i}a_je^{u_j}+r_i, \ \text{with} \ \mathbf{r}=P\mathbf{k}. 
\end{equation*}

\begin{remark} \label{remark1}
The parameters $a_i,\dots, a_n$ of \eqref{system} can be rescaled to  $a_1c_1,\dots,a_nc_n$, by using the transformation  $x_i\mapsto x_i/c_i$, for 
$c_i> 0$.  This linear transformation preserves the Poisson bracket and gives rise to 
an equivalent Hamiltonian system with Hamiltonian 
$H_y=\sum_{i=1}^n (a_i c_i y_i+k_i\log c_i y_i),$ in the new variables 
$y_i={x_i}/{c_i}$.  {For example, by setting $c_i=\frac{1}{\lvert a_i\rvert}$, all the nonzero $a_i$ can be rescaled to $1$ or $-1$.} Hence, we can consider systems with parameters  $a_i \in \{-1,0,1\}$ without any loss of generality. 
\end{remark}

 {In the present work we will restrict to the even-dimensional case; 
however, a similar approach can be considered for odd dimensions. 
Some additional comments on odd-dimensional cases as well as two examples, for $n=3$ and $n=5$, are given in the appendix.}

\section{Liouville integrability}\label{LI}
Following  \cite{KQV}, we introduce the functions $$v_i:=a_1 x_1+\cdots+a_i x_i, \ \ \  i=1,\dots,n.$$ 
If $a_1 a_2\dots a_n \neq 0$, the functions $v_i$ define new coordinates on $\mathbb{R}^n$ but generally this is not true.  
Furthermore, for any even $n$ we define the functions 
\begin{equation*} \label{JF}
  J_m(\mathbf{x})=\frac{x_1 x_3 \dots x_{2m-1}}{x_2 x_4 \dots x_{2m}}\;,  \ I_m(\mathbf{x})=\frac{x_{2m+2}x_{2m+4}\ldots x_n}{x_{2m+1} x_{2m+3}\ldots x_{n-1}}\;,
  \ F_m(\mathbf{x})=v_{2m} I_m(\mathbf{x})\;,
\end{equation*}
for  $m=1,\dots,\frac{n}{2}$. 
We also set $$H_0:=F_{n/2}=v_n=\sum_{i=1}^n a_i x_i \;,$$ 
 {which corresponds to} the Hamiltonian function in the case of $\mathbf{r}=\mathbf{k}=0$.  
So, the generic Hamiltonian of  \eqref{system} is written as  
$$H=H_0+\sum_{i=1}^n k_i \log x_i.$$
In \cite{KQV}, it is proved that 
  for any $m,l\in \{1,\dots,\frac{n}{2}\}$, 
  \begin{equation} \label{inv}
    \{J_m,J_l\}=\{F_m,F_l\}=\{F_m,H_0\}=0\;,
  \end{equation}
  as well as the following theorem which establishes the Liouville integrability of the system in the case of $\mathbf{r}=0$. 
\begin{theorem}\label{thm:integrable}
  Suppose that $n$ is even. Let $\ell$ denote the smallest integer such that $a_{\ell+1}\neq 0$ and let $\lambda:=\left[\frac{\ell}{2} \right]$. The $\frac{n}{2}$ functions
  $J_1,J_2,\dots,J_{\lambda},F_{\lambda+1},F_{\lambda+2},\dots,F_{\frac{n}{2}-1},H_0$ are pairwise in involution and
  functionally independent.
\end{theorem}

Here, our first goal is to determine the parameters $a_i$ and $\mathbf{r}$, so that the more general system \eqref{system} inherits the same integrals as the 
$\mathbf{r}=0$ case which ensure  Liouville integrability.  In the following, we always assume that $n$ is even.

\begin{lemma} \label{lem}
For any $m=1,\ldots, \frac{n}{2}-1$
\begin{equation} \label{Fint}
\{F_m,H\}=I_m \sum_{j=1}^{2m} k_j(v_j+v_{j-1}-v_{2m}) \;. 
\end{equation}
\end{lemma}

\begin{proof}
 {Since 
$\{F_m,H\}=\{F_m,H_0+\sum_{j=1}^n k_j \log x_j \}$ and  $\{F_m,H_0 \}=0$ (from \eqref{inv}), we get} 
\begin{equation} \label{prbrF}
\{F_m,H\}=\sum_{j=1}^n k_j \{F_m,  \log x_j\}=\sum_{j=1}^n k_j \{v_{2m}I_m,  \log x_j\}. 
\end{equation}
Now $\{v_{2m},\log x_j \}=\sum_{i=1}^{2m}\{x_i,x_j\} \frac{a_i}{x_j}$,  {so it follows that}   
\begin{equation*} 
\{v_{2m},\log x_j \}= \left\{
    \begin{array}{cc}
        v_{2m}\;, & \mbox{for } 2m<j\;,  \\
         v_j+v_{j-1}-v_{2m}\;,  & \mbox{for } j \leq 2m\;.
    \end{array}
    \right.
\end{equation*}
 {Also, we have} 
\begin{equation*} 
\{I_m,x_j \}= \left\{
    \begin{array}{cc}
        -I_m x_j\;, & \mbox{for } 2m<j\;,  \\
         0\;,  & \mbox{for } j \leq 2m\;
    \end{array}
    \right.
\end{equation*}
(this identity  {was proved} in \cite{KQV}) and  
\begin{equation*} 
\{I_m,\log x_j \}= \{I_m,x_j\} \frac{1}{x_j}=\left\{
    \begin{array}{cc}
        -I_m \;, & \mbox{for } 2m<j\;,  \\
         0\;,  & \mbox{for } j \leq 2m\;.
    \end{array}
    \right.
\end{equation*}
 {Therefore, we see that} 
\begin{equation*} 
\{v_{2m} I_m,\log x_j \}= \left\{
    \begin{array}{cc}
        0\;, & \mbox{for } 2m<j\;,  \\
         (v_j+v_{j-1}-v_{2m})I_m\;,  & \mbox{for } j \leq 2m\;,
    \end{array}
    \right.
\end{equation*}
and by substituting in \eqref{prbrF} we derive \eqref{Fint}. 
\end{proof}
We can recast the sum that appears in \eqref{Fint} to derive  
$$ \sum_{j=1}^{2m} k_j(v_j+v_{j-1}-v_{2m})= 
\sum_{j=1}^{2m}x_j a_j(-\sum_{i=1}^{j-1}k_i+\sum_{i=j+1}^{2m}k_i)\ ;$$
hence, from Lemma \ref{lem}, the next proposition follows. 
\begin{proposition} \label{propInt}
 {Suppose that $n$ is even.} For every $m=1,\ldots, \frac{n}{2}-1$,  
$\{F_m,H\}=0$ if and only if $S_{jm}=0$, for every $j=1,\ldots,2m$,  
where 
\begin{equation} \label{Sjm}
S_{jm}=a_j(-\sum_{i=1}^{j-1}k_i+\sum_{i=j+1}^{2m}k_i).
\end{equation}  
\end{proposition}

Solutions of the system 
$S_{jm}=0$, for $m=1,\dots,\frac{n}{2}-1$ and $j=1,\ldots,2m$, provide conditions on the parameters $a_i$ and $k_i$ ensuring that  the functions $F_{1},F_{2},\dots,F_{\frac{n}{2}-1}$ 
are first integrals of the system. Moreover, 
according to \eqref{inv},  these integrals are pairwise in involution. Therefore, in the case where $F_m \neq 0$, for all $m=1,\dots, \frac{n}{2}-1$,  these conditions on the parameters provide 
Liouville integrability.   
For example, in the particular case where $a_j \neq 0$, for every $j=1,\dots, n$, the corresponding system implies the unique solution $k_1=k_2=\dots=k_{n-2}=0$. 

\begin{corollary}
For $a_1 a_2 \dots a_n \neq 0$, $k_1=k_2=\dots=k_{n-2}=0$, $k_{n-1},k_{n} \in \mathbb{R}$ and $\mathbf{r}=P\mathbf{k}=(k_{n-1}+k_n, k_{n-1}+k_n, \dots, k_n,-k_{n-1})$, the Hamiltonian system \eqref{system} is Liouville integrable 
with first integrals $H,F_{1},F_{2},\dots,F_{\frac{n}{2}-1}${\footnote{The proof of the functional independence of the integrals is given in Prop. \ref{indip}.}}. 
\end{corollary}
 
Now, let  $a_1=a_2=\dots=a_{\ell}=0$, $a_{\ell+1}\neq 0$ and $\lambda:=\left[\frac{\ell}{2} \right]$. In such a case, $F_1=\dots=F_{\lambda}=0$. So, for any choice of parameters there are not enough $F$-type integrals to ensure the integrability of the system. However, Theorem \ref{thm:integrable} suggests that we could probably replace the first $\lambda$ missing $F$-integrals by $\lambda$ $J$-integrals. Hence next, we are going to determine the conditions on the parameters to ensure that $\{J_m,H\}=0$, for $m=1,\dots, \lambda$ 
{\footnote{For $m>\lambda$,  $J_m$ cannot be an integral of the system i.e. $\{J_m,H\}\neq0$. So, the total number of $F$ and $J$ integrals cannot exceed $\frac{n}{2}-1$.}}. 

\begin{lemma} \label{lemma2}
Let  $a_1=a_2=\dots=a_{\ell}=0$, $a_{\ell+1}\neq 0$ and $\lambda:=\left[\frac{\ell}{2} \right]$. 
Then, 
\begin{equation} \label{lemJInt}
\{J_m,H\}=J_m(k_1+k_2+\dots+k_{2m}),
\end{equation}
for $m=1,\dots,\lambda$.
\end{lemma}

\begin{proof}
We consider $m \in \{1,\dots,\lambda\}.$ From Theorem \ref{thm:integrable}, it follows that 
$\{J_m,H_0\}=0$.  
So, 
\begin{equation} \label{lem2JH}
\{J_m,H\}=\sum_{j=1}^n k_j \{J_m,  \log x_j\}=\sum_{j=1}^{2m} k_j \{J_m,  \log x_j\}+\sum_{j=2m+1}^n k_j \{J_m,  \log x_j\}.
\end{equation}
Also, 
$$\{J_m,x_j\}=\sum_{i=1}^n\{x_i,x_j\} \frac{\partial J_m}{\partial x_i}=
\sum_{i=1}^{j-1}x_i x_j \frac{\partial J_m}{\partial x_i}-\sum_{i=j+1}^{n}x_i x_j \frac{\partial J_m}{\partial x_i}$$
and 
\begin{equation*} 
x_i \frac{\partial J_m}{\partial x_i}= \left\{
    \begin{array}{cc}
        (-1)^{i+1} J_m\;, & \mbox{for } 1 \leq i \leq 2m\;,  \\
         0\;,  & \mbox{for } 2m<i \leq n\;.
    \end{array}
    \right.
\end{equation*}
Consequently, after some calculations we obtain 
\begin{equation*} 
\{J_m,x_j \}= \left\{
    \begin{array}{cc}
        J_m x_j\;, & \mbox{for } j \leq 2m\;,  \\
         0\;,  & \mbox{for } j > 2m\; 
    \end{array}
    \right.
\  \text{and} \  \ 
\{J_m,\log x_j \}=\left\{
    \begin{array}{cc}
        J_m\;, & \mbox{for } j \leq 2m\;,  \\
         0\;,  & \mbox{for } j > 2m\;. 
    \end{array}
    \right.
\end{equation*}
Substituting this  {into} \eqref{lem2JH}, we derive \eqref{lemJInt}. 

\end{proof}
 
 Finally, if we combine Lemma \ref{lemma2} with Prop.~\ref{propInt} we come up with the following 
 theorem.
 
 \begin{theorem} \label{mainthm}
  { Suppose that $n$ is even.}  Let $\ell$ denote the smallest integer such that $a_{\ell+1}\neq 0$ and let $\lambda=\left[\frac{\ell}{2} \right]$. The $\frac{n}{2}$ functions
  $J_1,J_2,\dots,J_{\lambda},F_{\lambda+1},F_{\lambda+2},\dots,F_{\frac{n}{2}-1},H$ are pairwise in involution if and only if 
$k_{2i}=-k_{2i-1}$, for $i=1,\dots, \lambda$, and $S_{jm}=0$, for $m=\lambda+1,\dots,\frac{n}{2}-1$, $j=\ell+1,\ldots,2m$.
 \end{theorem}

\begin{proof}
Let  $a_1=a_2=\dots=a_{\ell}=0$, $a_{\ell+1}\neq 0$ and $\lambda=\left[\frac{\ell}{2} \right]$. 
From Lemma \ref{lemma2}, we conclude that $\{J_1,H\}=\{J_2,H\}=\dots=\{J_{\lambda},H\}=0$ if and only if $k_1+k_2+\dots+k_{2i}=0$, for all $i=1,\dots, \lambda$, which is equivalent to $k_{2i}=-k_{2i-1}$, for $i=1,\dots, \lambda$. 
Also, from Prop. \ref{propInt}, we derive that for $m=\lambda+1,\dots,\frac{n}{2}-1$, 
$\{F_m,H\}=0$ if and only if $S_{jm}=0$, for every $j=\ell+1,\ldots,2m$ 
(for $j=1,\dots \ell$, $S_{jm}=0$, since $a_1=\dots=a_l=0$). Finally,  {Theorem \ref{thm:integrable}} shows that  {all the other} pairs of functions are in involution too.
\end{proof}

We will close this section by proving the functional independence of the integrals. 
\begin{proposition} \label{indip}
 {For every even $n$,} the functions  $J_1,J_2,\dots,J_{\lambda},F_{\lambda+1},F_{\lambda+2},\dots,F_{\frac{n}{2}-1}$,\\
$H$, are functionally
independent. 
\end{proposition}
\begin{proof}
For $\mathbf{k}=0$,  $J_1,\dots,J_{\lambda},F_{\lambda+1},\dots,F_{\frac{n}{2}-1}$,$H$ are functionally independent. This follows from Theorem \ref{thm:integrable}, since in this case 
$H$ coincides with $H_0$. Hence, by continuity the same functions 
remain functionally independent for parameters 
$\mathbf{k}$ in a  sufficiently small open neighborhood $U$ of $\mathbf{k}=0$. Now, let us consider 
any $\mathbf{k}=(k_1,\dots,k_n) \in \mathbb{R}^n$. Then there is $\mu > 0$ and 
$\mathbf{k'}=(k_1'\dots,k_n') \in U$, such that $\mathbf{k}=\mu \mathbf{k}'$. Also,  {in} view of Remark \ref{remark1}, 
we can rescale the $a_i$ parameters to $\mu a_i$, by setting $y_i=x_i/\mu$. The Hamiltonian function in the new $y$-coordinates then becomes 
$$H_y(\mathbf{y})=\sum_{i=1}^n (a_i \mu y_i+k_i\log \mu y_i)=\mu\sum_{i=1}^n (a_i y_i+k'_i\log \mu y_i).$$
So, $dH_y=\mu dH'$, where $H'(\mathbf{y})=\sum_{i=1}^n (a_i y_i+k_i'\log y_i)$, i.e. the  Hamiltonian of the corresponding system with parameters $a_i$ and $k'_i$. 
Therefore, from the functional independence of  $J_1,\dots,J_{\lambda},F_{\lambda+1},\dots,F_{\frac{n}{2}-1},H'$ that we proved, the functional independence of  $J_1,\dots,J_{\lambda},F_{\lambda+1},\dots,F_{\frac{n}{2}-1},H_y$ follows and consequently the functional independence of $J_1,\dots,J_{\lambda},F_{\lambda+1},\dots,F_{\frac{n}{2}-1},H$ 
for all parameters $a_i$ and $k_i$. 

\end{proof}

\section{Symmetry and superintegrability}

In \cite{KKQTV,KQV}, a second set of first integrals in involution has been introduced for the case of 
$\mathbf{r}=0$. By considering this set of integrals we can derive more integrable cases of our system. 
The main observation to accomplish this is that system \eqref{system2} remains invariant under the transformation 
$x_i \mapsto x_{n+1-i}$ and the reparametrization $a_i \mapsto -a_{n+1-i}$, $k_i \mapsto -k_{n+1-i}$, for $i=1,\dots, n$.
Let us now consider the involution $\iota(x_1,x_2,\dots,x_n)\mapsto(x_n,x_{n-1},\dots,x_1)$
and the functions 
\begin{equation*} 
  \tilde{J}_m=J \circ \iota \;,  \ \tilde{I}_m=I \circ \iota \;\;,
  \ \tilde{F}_m=\tilde{v}_{2m} \tilde{I}_m\;,
\end{equation*}
where $\tilde{v}_i:=a_{n+1-i} x_{n+1-i}+a_{n+2-i} x_{n+2-i}+\cdots+a_n x_n$,  
for  $i=1,\dots,n$. Then, by Theorem \ref{mainthm} and the described symmetry of the system  we derive the next  {theorem}.

 \begin{theorem} \label{Revthm}
  {Suppose that $n$ is even.} Let $d$ denote the smallest integer such that $a_{n-d}\neq 0$ and let $\delta=\left[\frac{d}{2} \right]$. The $\frac{n}{2}$ functions
  $\tilde{J}_1,\tilde{J}_2,\dots,\tilde{J}_{\delta},\tilde{F}_{\delta+1},\tilde{F}_{\delta+2},\dots,\tilde{F}_{\frac{n}{2}-1},H$ are pairwise in involution if and only if 
$k_{n+1-2i}=-k_{n+2-2i}$, for $i=1,\dots, \delta$, and $\tilde{S}_{jm}=0$, for $m=\delta+1,\dots,\frac{n}{2}-1$, $j=d+1,\ldots,2m$, where
\begin{equation*} 
\tilde{S}_{jm}=a_{n+1-j}(-\sum_{i=1}^{j-1}k_{n+1-i}+\sum_{i=j+1}^{2m}k_{n+1-i})\;.
\end{equation*}  
 \end{theorem}

Theorem \ref{Revthm}, determines different values of the parameters of the system that lead to integrability. Furthermore,  a combination of Theorems \ref{mainthm}-\ref{Revthm} provide some superintegrable cases. 
For any $\ell,d \in \{0,1,\dots, n-1\}$, we consider the following two sets of parameters: 
\begin{align*}
\Sigma_{\ell}=\{&(\mathbf{a},\mathbf{k}) \in \mathbb{R}^{2n}:a_1=a_2=\dots a_{\ell}=0, a_{\ell+1} \neq 0, k_{2i}+k_{2i-1}=S_{jm}=0, \\ 
&\text{for} \ 
 i=1,\dots, \left[\frac{\ell}{2} \right], 
\, m=\left[\frac{\ell}{2} \right]+1,\dots,\frac{n}{2}-1, \ j=\ell+1,\dots,2m \ \}, \\ 
\tilde{\Sigma}_{d}=\{&(\mathbf{a},\mathbf{k}) \in \mathbb{R}^{2n} :a_n=a_{n-1}=\dots=a_{n-d+1}=0,a_{n-d}\neq0, \\ 
& k_{n+1-2i}+k_{n+2-2i}=\tilde{S}_{jm}=0,  
 \text{for} \  i=1,\dots, \left[\frac{d}{2} \right], 
\, m=\left[\frac{d}{2} \right]+1,\dots,\frac{n}{2}-1, \\
& j=d+1,...,2m \ \},
\end{align*}
where $(\mathbf{a},\mathbf{k}):=(a_1,\dots,a_n,k_1,\dots,k_n)$.  
Then we conclude with the following theorem. 

\begin{theorem} \label{lastthm}
If $(\mathbf{a},\mathbf{k}) \in \Sigma_{\ell} \cup \tilde{\Sigma}_{d}$ for some $\ell,d \in \{0,1,\dots, n-1\}$, then  {for every even $n$} system \eqref{system2} with parameters $\mathbf{a},\mathbf{k}$ is Liouville integrable. If 
$ {(\mathbf{a},\mathbf{k})} \in \Sigma_{\ell} \cap \tilde{\Sigma}_{d}$, then the corresponding system \eqref{system2} is superintegrable, i.e. it admits the following $n-1$ functionally independent integrals:\\ $J_1,J_2,\dots,J_{\lambda},F_{\lambda+1},F_{\lambda+2},\dots,F_{\frac{n}{2}-1},
\tilde{J}_1,\tilde{J}_2,\dots,\tilde{J}_{\delta},\tilde{F}_{\delta+1},\tilde{F}_{\delta+2},\dots,\tilde{F}_{\frac{n}{2}-1},H$. 
\end{theorem}

\begin{example}
The simplest interesting case is $n=4$ 
(for $n=2$ the system is always integrable since it is Hamiltonian). In this case we have, 
\begin{align*}
&\Sigma_0=\{(\mathbf{a},\mathbf{k}) \in \mathbb{R}^{8}:a_1\neq0,k_2=a_2 k_1=0 \}, \\   
&\Sigma_1=\{(\mathbf{a},\mathbf{k}) \in \mathbb{R}^{8}:a_1=0,a_2\neq0,k_1=0 \}, \\ 
&\Sigma_2=\{(\mathbf{a},\mathbf{k}) \in \mathbb{R}^{8}:a_1=a_2=0,a_3\neq0,k_2=-k_1 \},   \\
&\Sigma_3=\{(\mathbf{a},\mathbf{k}) \in \mathbb{R}^{8}:a_1=a_2=a_3=0,a_4\neq0,k_2=-k_1 \}, \\
&\tilde{\Sigma}_0=\{(\mathbf{a},\mathbf{k}) \in \mathbb{R}^{8}:a_4\neq0,k_3=a_3 k_4=0 \}, \\ 
&\tilde{\Sigma}_1=\{(\mathbf{a},\mathbf{k}) \in \mathbb{R}^{8}:a_4=0,a_3\neq0,k_4=0 \}, \\
&\tilde{\Sigma}_2=\{(\mathbf{a},\mathbf{k}) \in \mathbb{R}^{8}:a_4=a_3=0,a_2\neq0,k_4=-k_3 \},\\
&\tilde{\Sigma}_3=\{(\mathbf{a},\mathbf{k}) \in \mathbb{R}^{8}:
a_4=a_3=a_2=0,a_1\neq0,k_4=-k_3 \}, 
\end{align*}
where $(\mathbf{a},\mathbf{k}) =(a_1,\dots,a_4,k_1,\dots, k_4)$. 
Now, using Theorem \ref{lastthm} we can detect different integrable and superintegrable cases. 
So for example, when $a_1\dots a_n \neq 0$,  from $\Sigma_{0} \cup \tilde{\Sigma}_{0}$ 
we come up with two integrable cases,  {for $\mathbf{k}=(0,0,k_3,k_4)$ and $\mathbf{k}=(k_1,k_2,0,0)$}, while 
the only superintegrable case 
that is derived from $\Sigma_{0} \cap \tilde{\Sigma}_{0}$  is when  {$\mathbf{k}=0$}. 
On the other hand, for $a_2=0$ and $a_1 a_3 a_4 \neq 0$, we derive the integrable cases with 
 {
$\mathbf{k}=(k_1,0,k_3,k_4)$ and $\mathbf{k}=(k_1,k_2,0,0)$},
and the superintegrable case for  {
$\mathbf{k}=(k_1,0,0,0)$}.  Proceeding in this way, we can detect all the integrable and superintegrable cases 
given by Theorem \ref{lastthm}.  

\end{example} 

\begin{figure}
\centering
\includegraphics[scale=0.25]{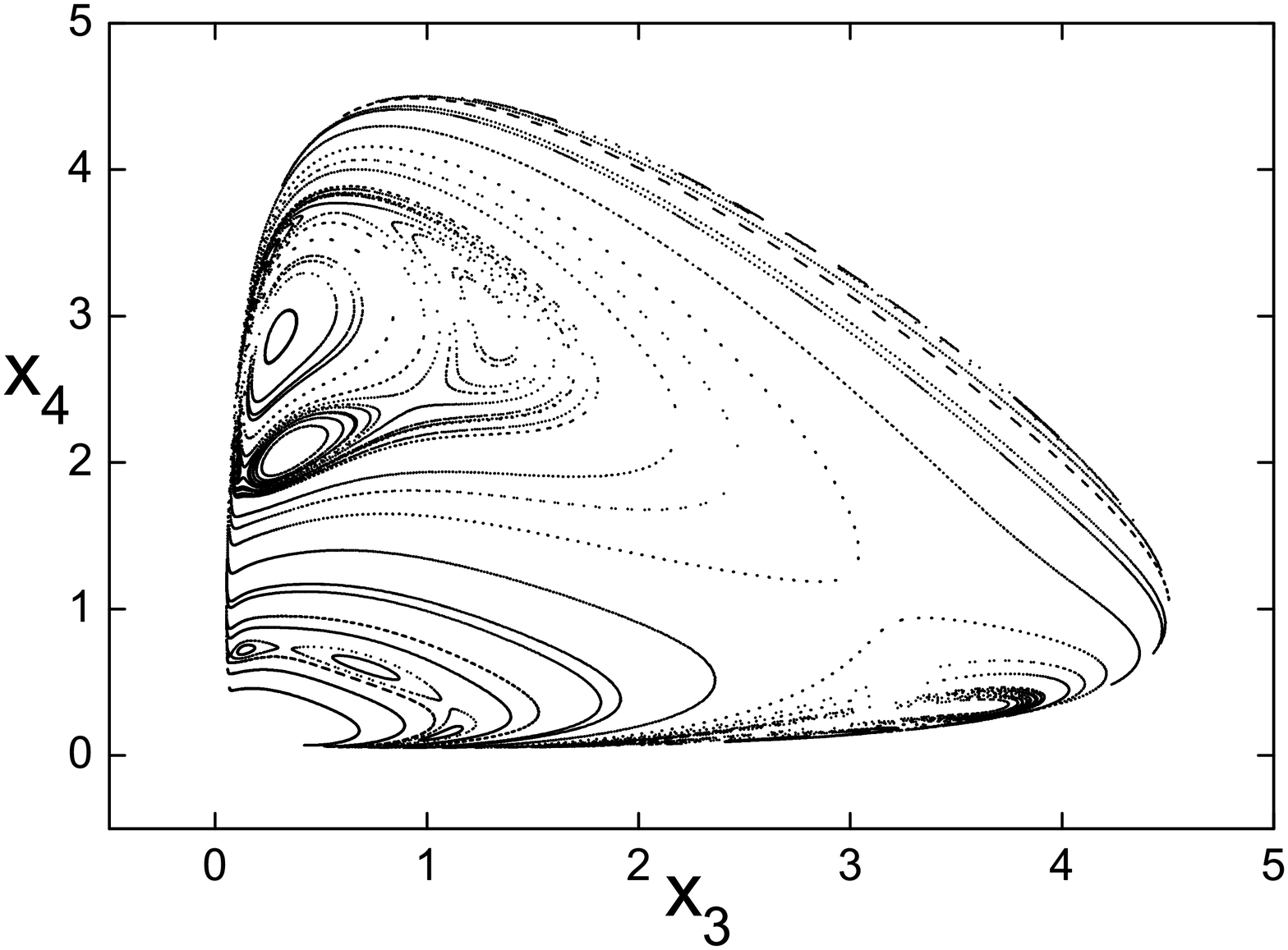}
\includegraphics[scale=0.25]{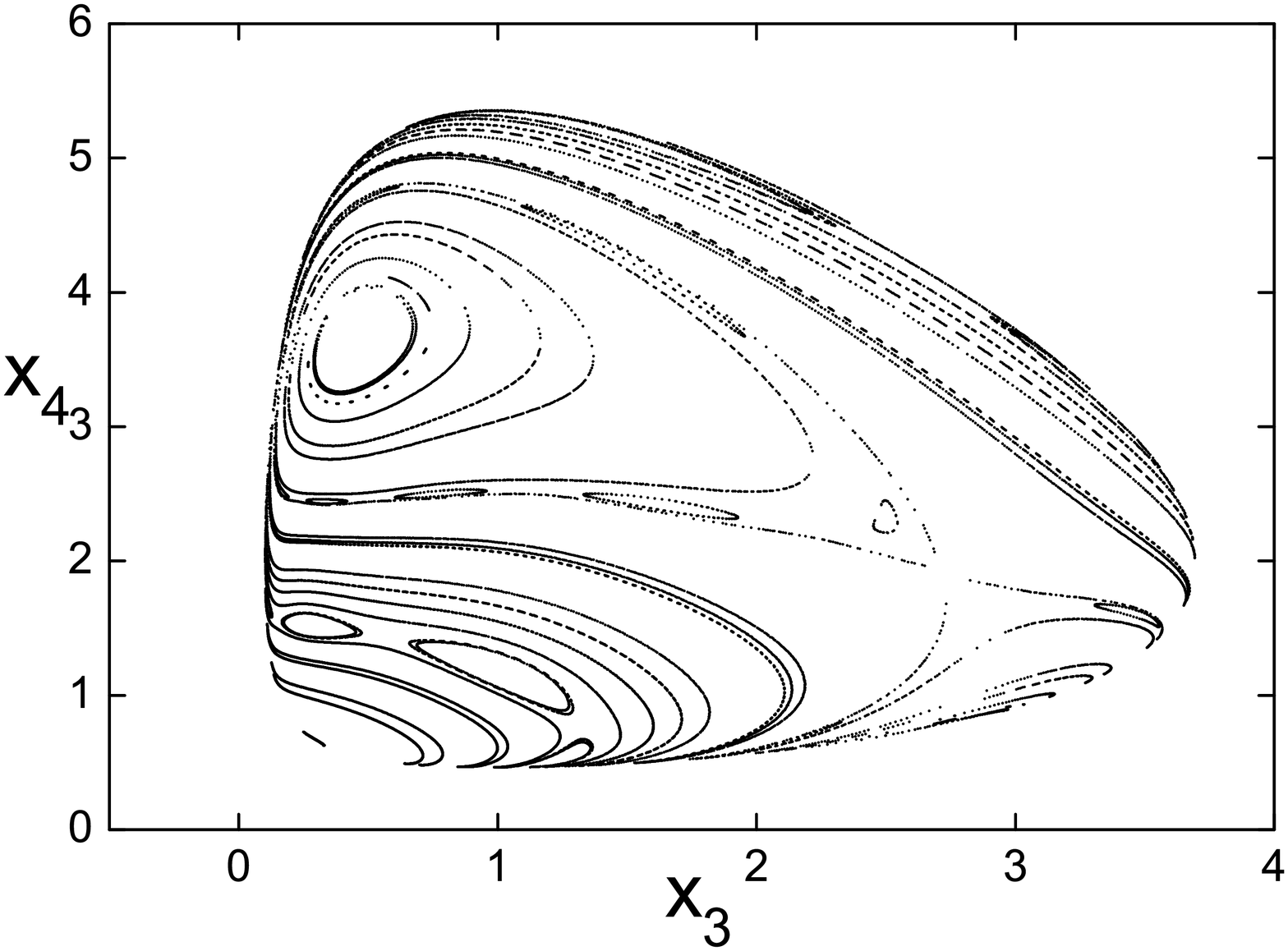} \\ 
\includegraphics[scale=0.25]{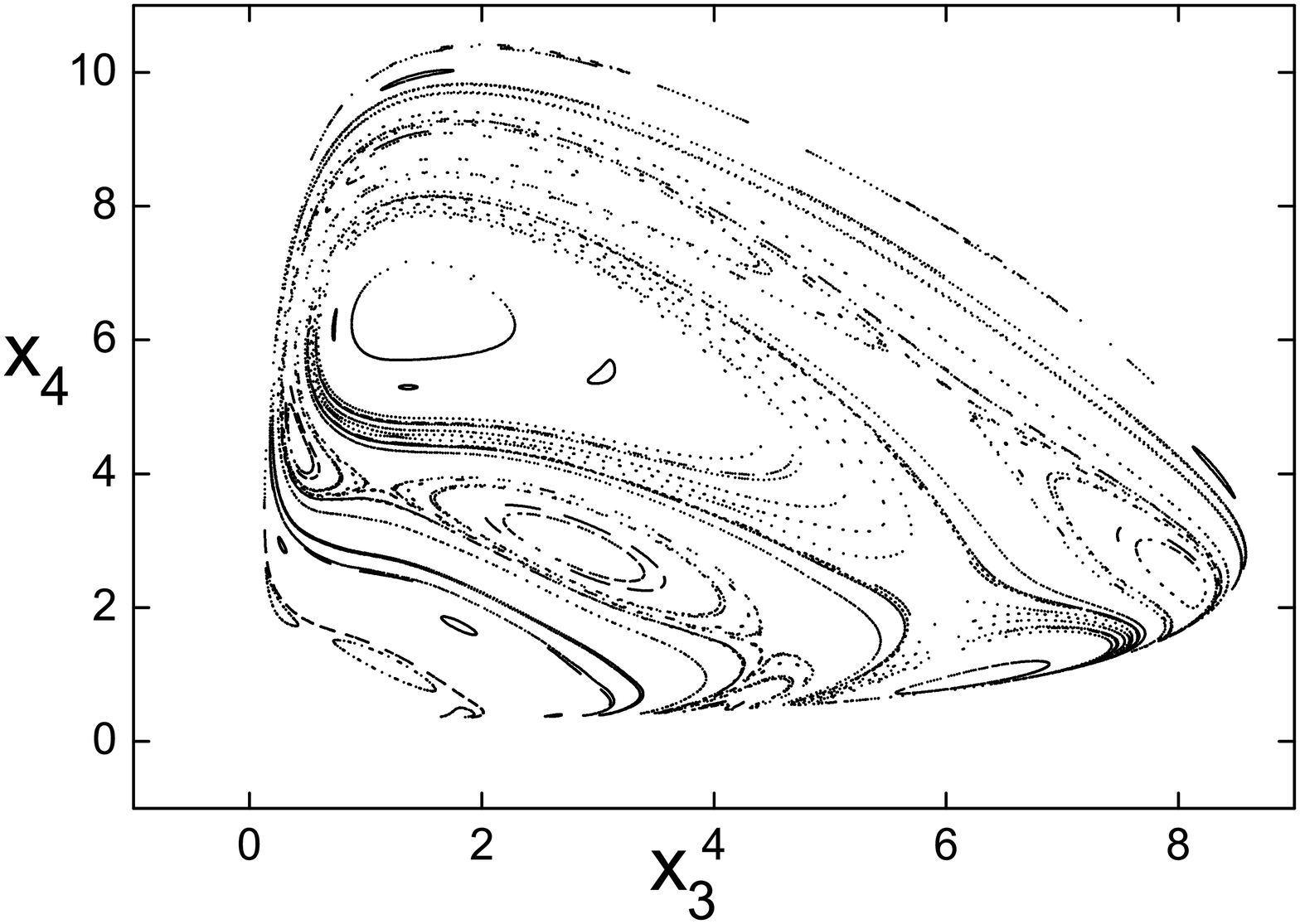}
\includegraphics[scale=0.25]{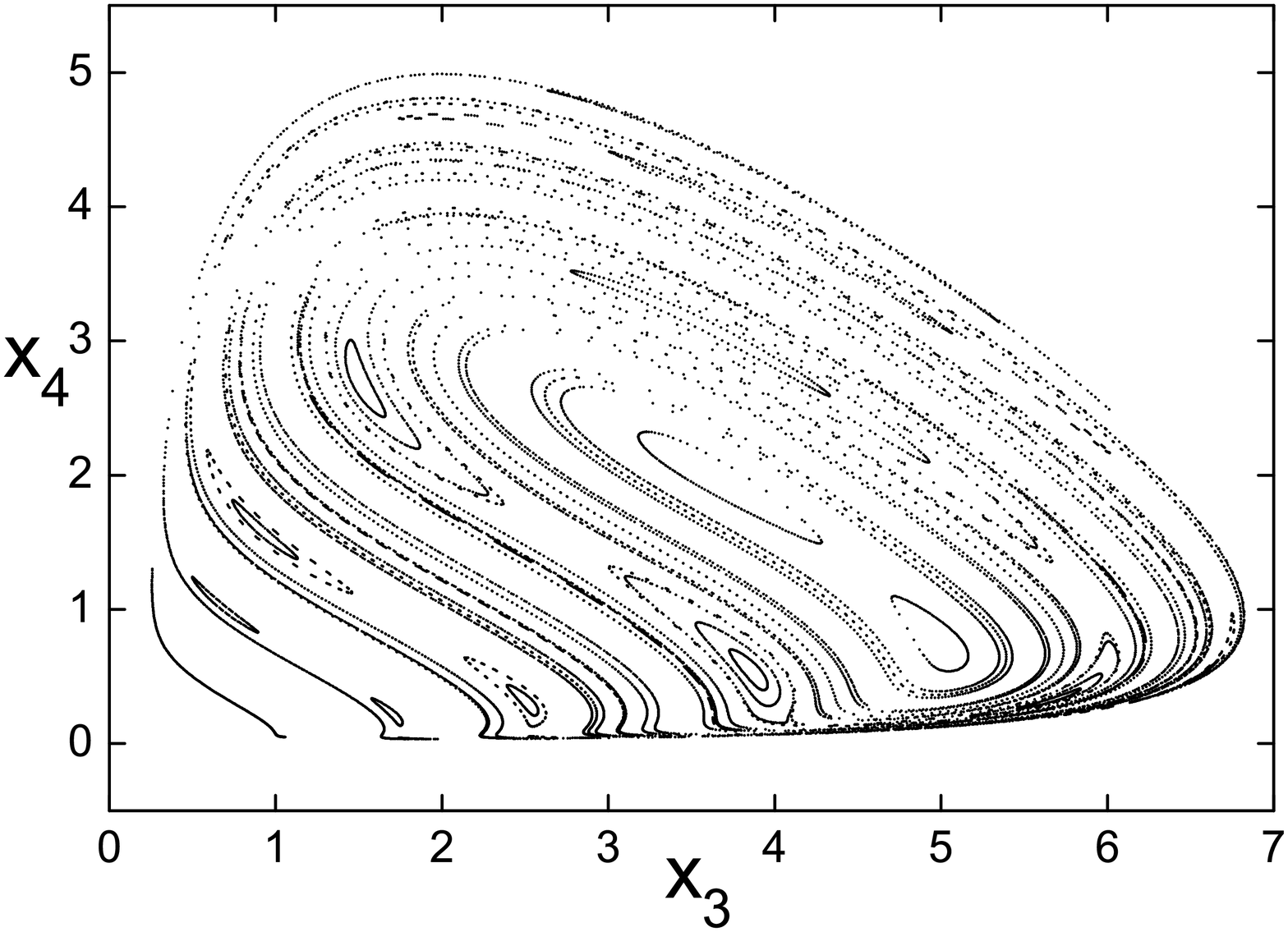}
\caption{ The Poincar{\' e} surface of section  {$x_2=1,x_1>1$} for the Lotka--Volterra system 
with $a_i=1$ and $E=6$ for various $k_i$, $i=1,2,3,4$ values: (a) $(k_1,k_2, k_3, k_4) =(-1, -2,- 1,-1)$ 
(b) $(k_1,k_2, k_3, k_4) =(-1, -2,- 1,-2)$, (c) $(k_1,k_2, k_3, k_4) =(-1,-4,-2,-3)$, (d) 
$(k_1,k_2, k_3, k_4) =(-1, -4,- 2,-1)$.   
\label{fig2}  }
\end{figure}
\section{Numerical results for $n=4$}

The purpose of this section is to explore numerically the behavior of $4$-dimensional Lotka--Volterra systems of the form \eqref{system} and investigate their integrability in 
cases that are not described in the previous sections. 
In the rest of the paper we will restrict to the case of $a_1=a_2=a_3=a_4=1$ and we 
vary only the $k_i$ values.
We perform a series of numerical calculations for the system 
\begin{equation}\label{eqss}
  \dot{x_i}=x_i \left(\sum_{j=1}^ { {4}} P_{ij}( x_j+k_j) \right)\;, \ \ i=1,\dots,4,
\end{equation}

\begin{figure}
\centering
\includegraphics[scale=0.25]{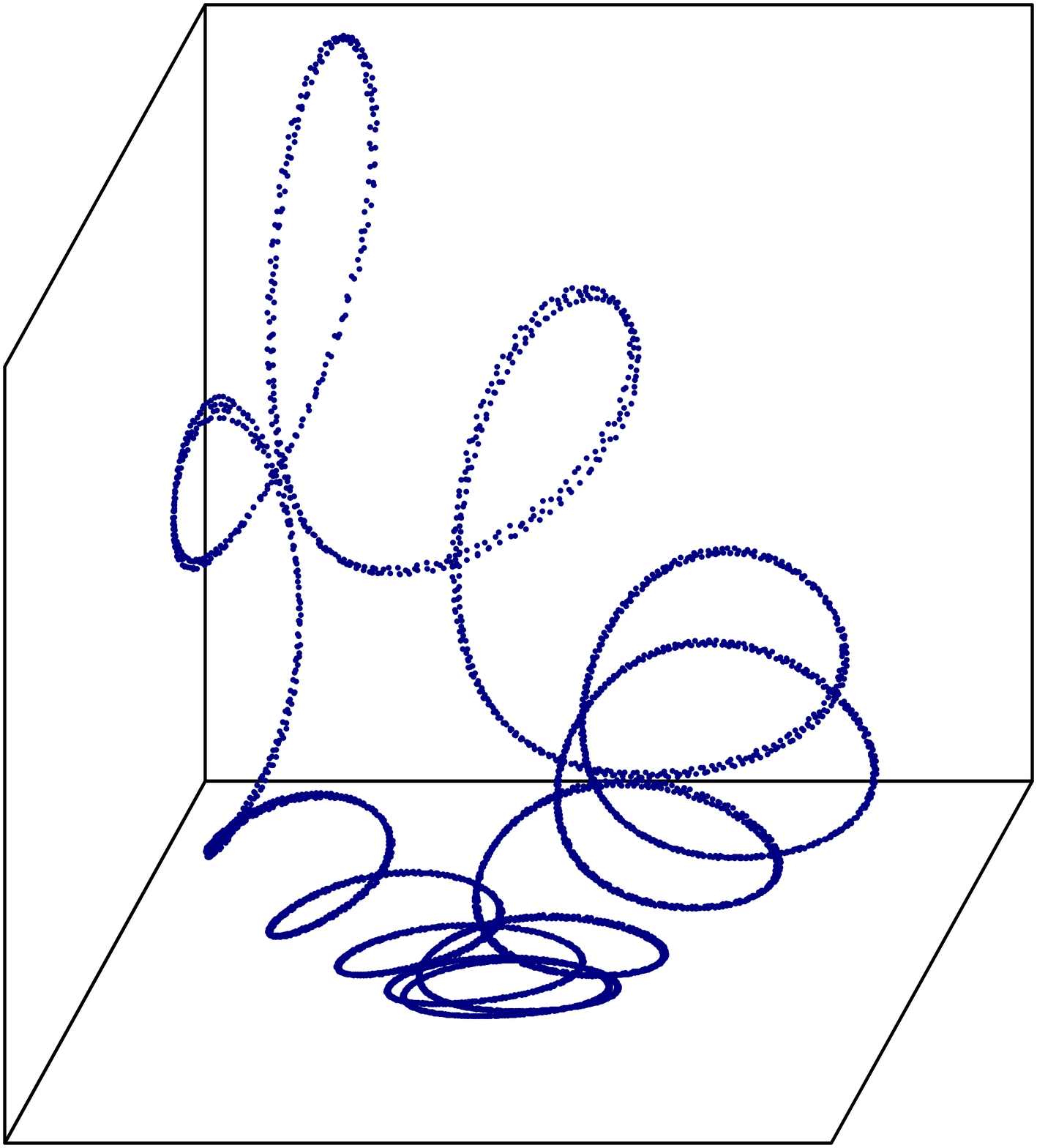}
\includegraphics[scale=0.25]{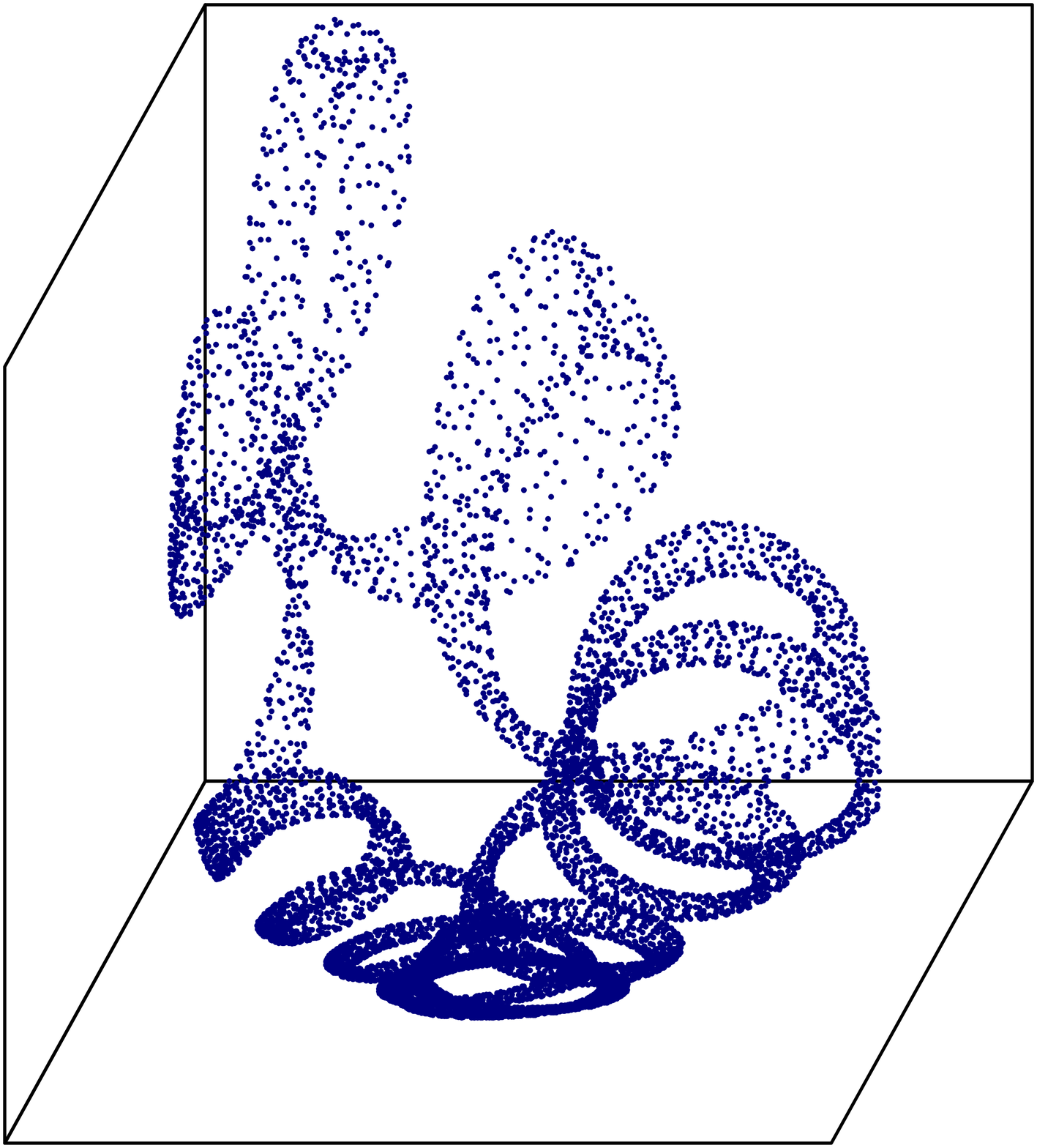} \\ 
\includegraphics[scale=0.22]{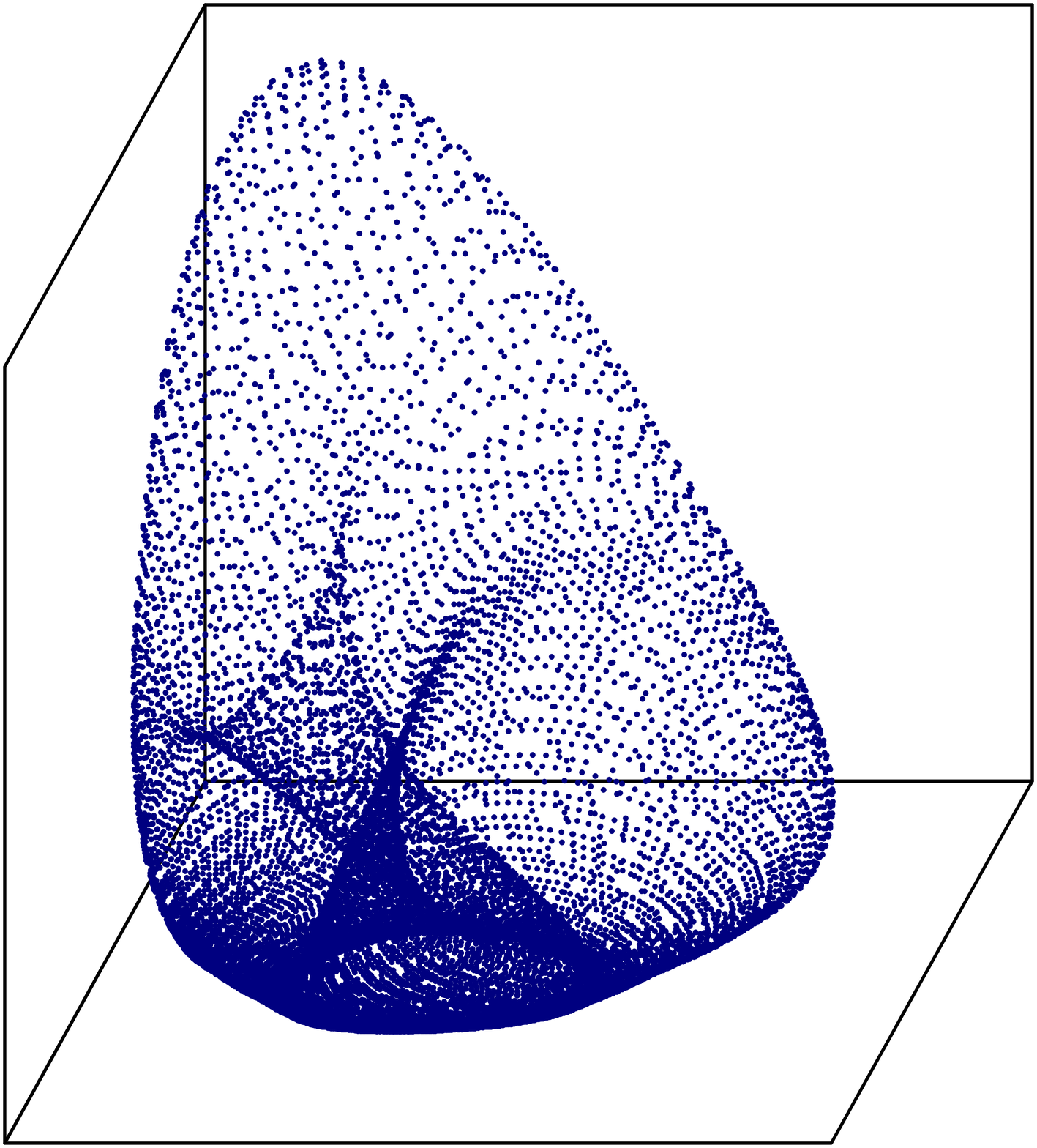}
\includegraphics[scale=0.22]{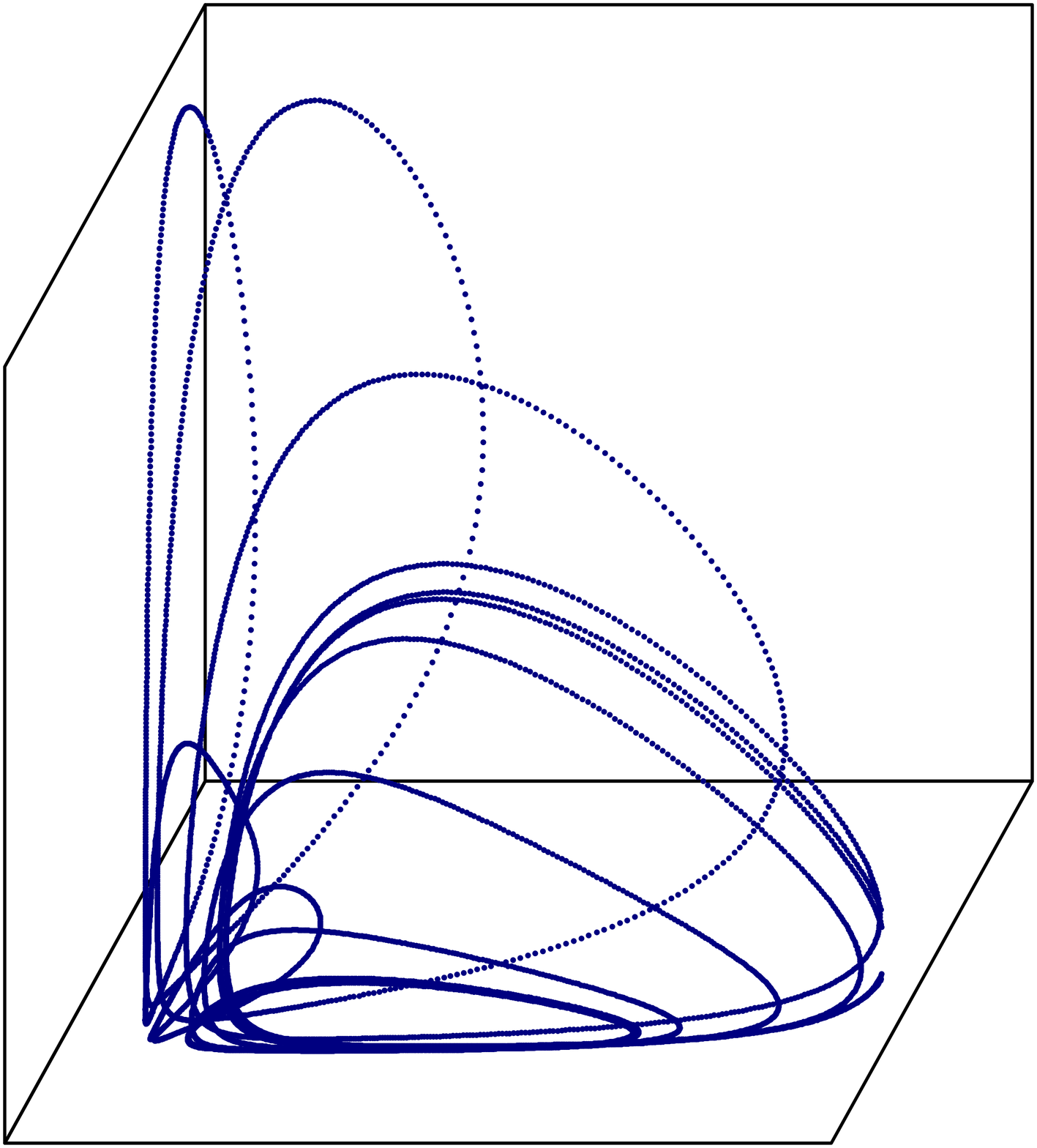}
\caption{ 3D projections on the $x_2,x_3,x_4$ plane for the system with $(k_1,k_2, k_3, k_4) =(-1, -4,- 2,-1)$
and for initial conditions:  (a) close to a fixed point of Fig.\ref{fig2}(d) ($E=6$), 
(b) on an ellipse around the fixed point ($E=6$), (c) randomly chosen from Fig.\ref{fig2}(d) ($E=6$)
and (d) randomly chosen at a higher total energy ($E=20$) exhibiting chaotic behavior.  
\label{fig3}}
\end{figure}

with different  $k_1,\ldots,k_4$ values, which are
complementary to the two integrable cases described by Theorem \ref{lastthm}.
We numerically integrate the system's equations of motion together with its 
variational equations to compute the value of the largest
Lyapunov exponent $\lambda $. The variational equations of the system  (\ref{eqss}) are
\begin{eqnarray}\label{tang}
\delta \dot{ \mathbf{x}}= [J \cdot \nabla ^2 H (\mathbf{x}(t))] \cdot  \delta \mathbf{x}~~,
\end{eqnarray}
where $\delta \mathbf{x}=(\delta x_1,\delta x_2,\delta x_3,\delta x_4)$ is a vector which evolves on the tangent space 
of the system (\ref{eqss}) and $\nabla ^2 H$ denotes the Hessian matrix of the Hamiltonian function $H$
calculated along the reference orbit $\mathbf{x}(t)$ of the system (\ref{eqss}).
In particular, we used the classical Runge--Kutta forth-order scheme with time-step $\tau =10^{-4}$
for the numerical integration of the systems  (\ref{eqss}) and (\ref{tang}), which 
conserved the energy  {$E=H(\mathbf{x})$} of the system (\ref{eqss}) with accuracy of more than 
8 significant figures during integration times of the order of a few thousand.
The indicator which controls of the relative energy error is 
$$RE=\log_{10} \left|  \frac{ E(t)-E_0 }{E_0} \right| ~~,$$
 where $E_0$ is the initial energy of the system and $E(t)$ the actual energy during the
 numerical integration.

For $k_i<0$, $i=1,\dots4$, the point  $\mathbf{x}_0=(-k_1,-k_2,-k_3,-k_4)$ is an elliptic
fixed point of the system. Furthermore, in this case $H(\mathbf{x})$ admits 
a global minimum at $\mathbf{x}_0$ and all the orbits of the system are bounded.

We start our numerical study with examples of bounded motion, which correspond to
negative values for  {all $k_i$}.
 In Fig.\ref{fig2} some Poincar{\' e} surfaces of section  $x_2=1$, $x_1>1$ are shown for different $k_i<0$ values at $E=6$.  
However, at this energy level all of them exhibit regular behavior.
These Poincar{\' e} surfaces of section are constructed for a grid of initial conditions on the 
$x_3, x_4$ plane, with $x_2=1$ and $x_1$ found numerically by Newton's method requiring 
that $H(\mathbf{x})=E$. We find a rich morphology 
consisting of periodic and quasiperiodic trajectories, island chains as well as separatrices. 
Each fixed point on the Poincar{\' e} surface represents a periodic orbit, while the ellipse-like 
curves correspond to quasiperiodic trajectories lying on tori. 
Fig.\ref{fig3}
presents different trajectories projected on the 
$x_2,x_3,x_4$ plane for the system with $(k_1,k_2, k_3, k_4) =(-1, -4,- 2,-1)$
which corresponds to Fig.\ref{fig2}(d). 
 {The first three panels of Fig.\ref{fig3} correspond to $E=6$
and the last one to $E=20$.}

\begin{figure}
\centering
\includegraphics[scale=0.25]{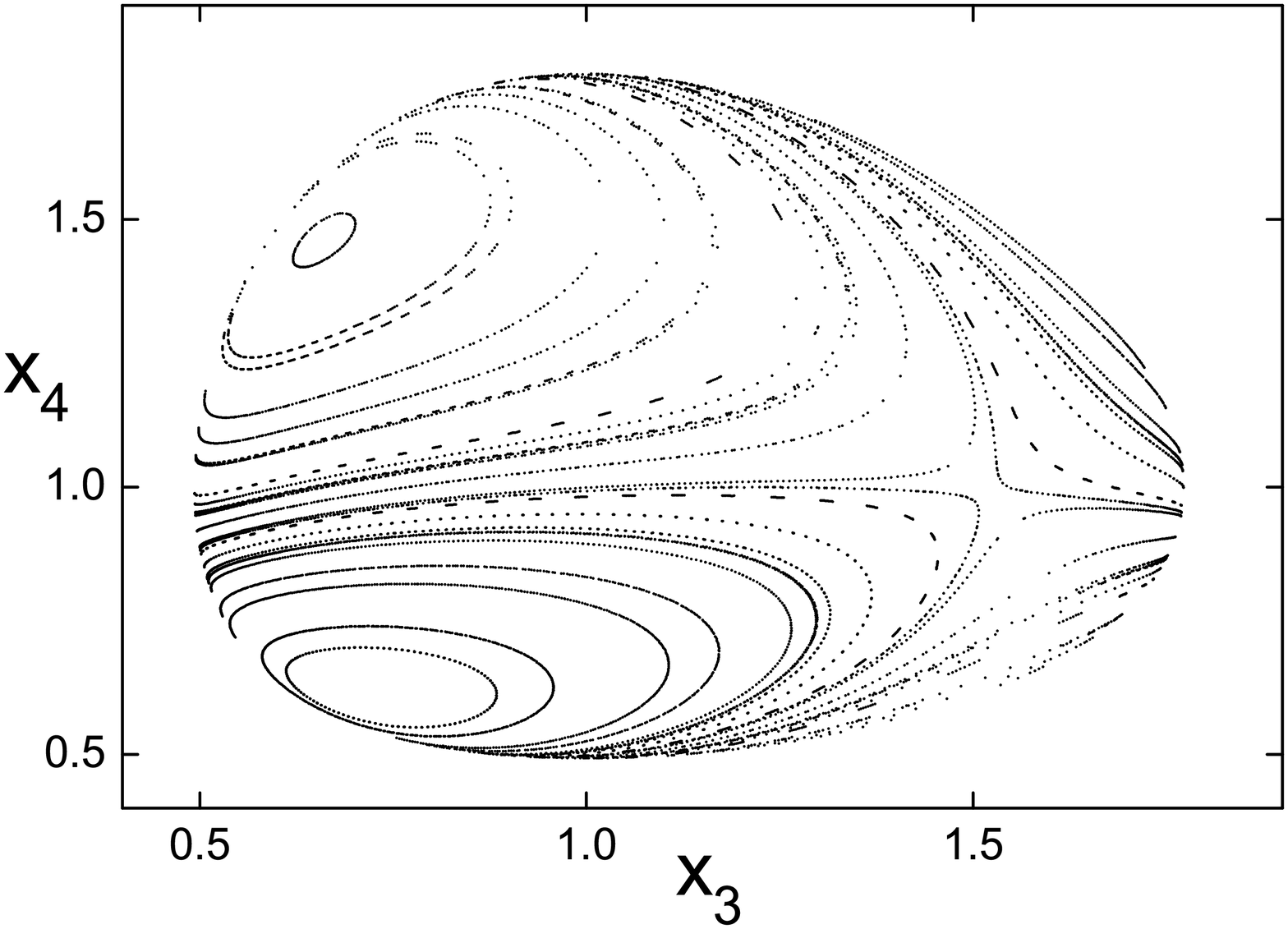}
\includegraphics[scale=0.25]{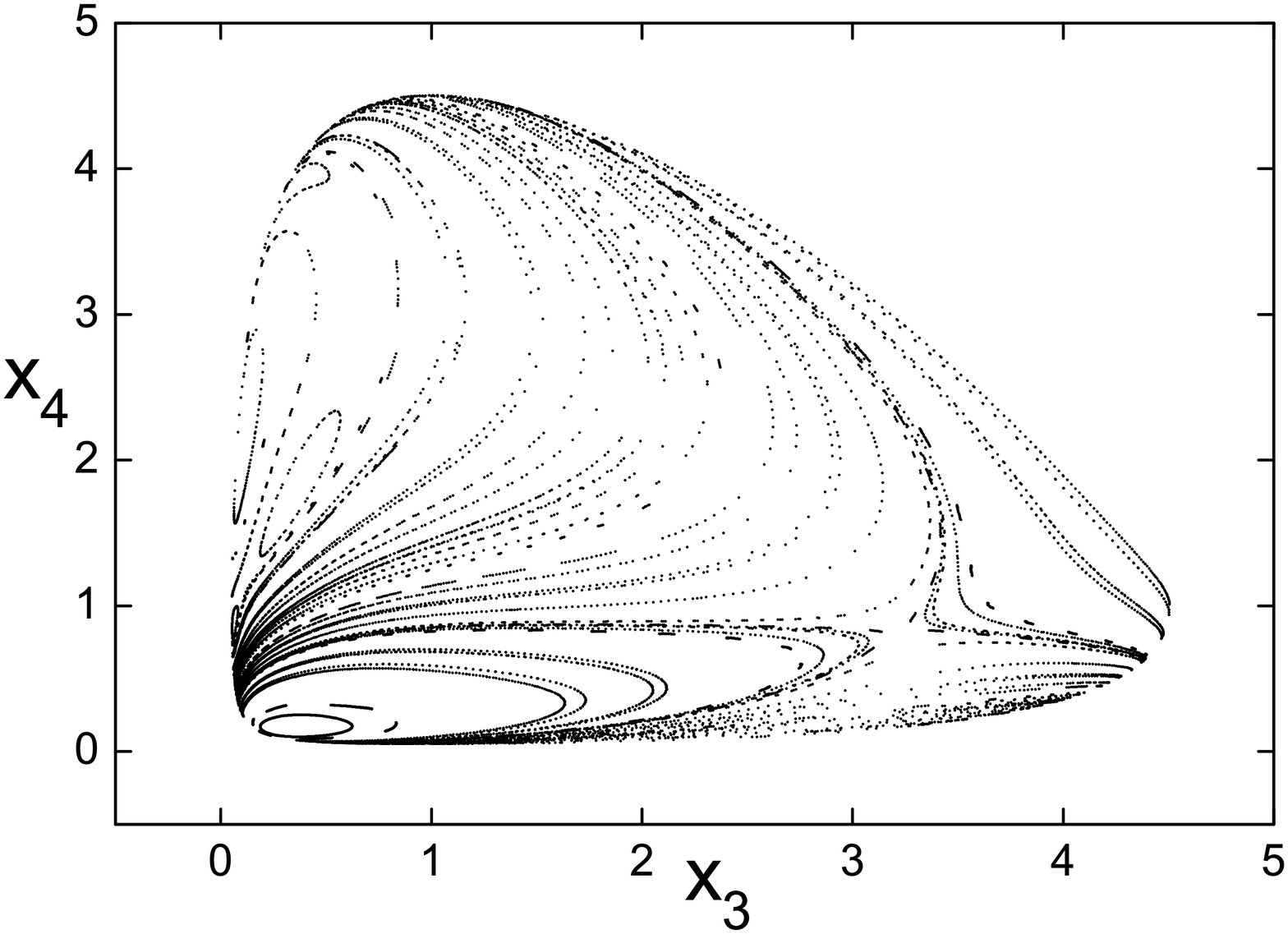} \\ 
\includegraphics[scale=0.25]{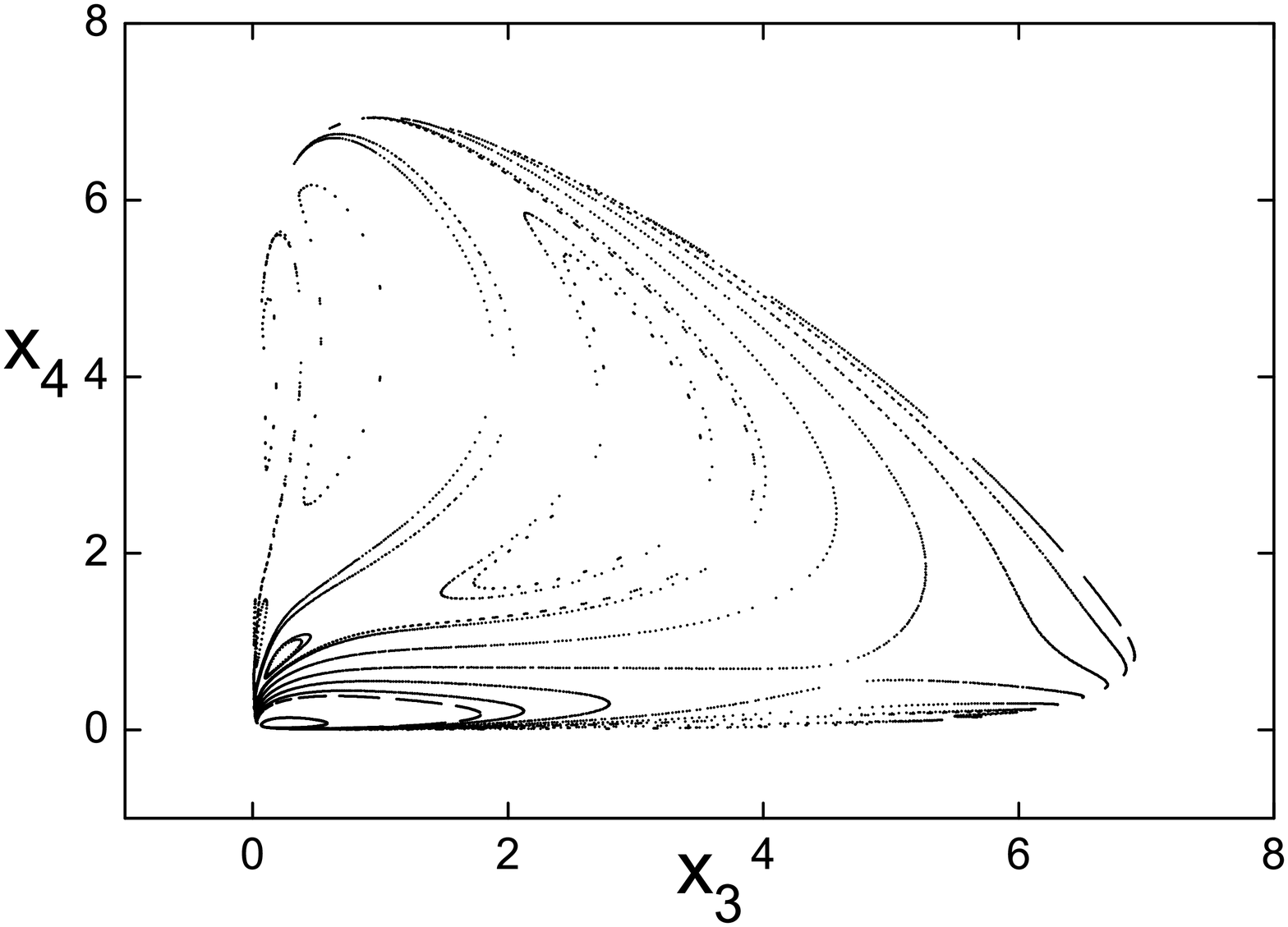}
\includegraphics[scale=0.25]{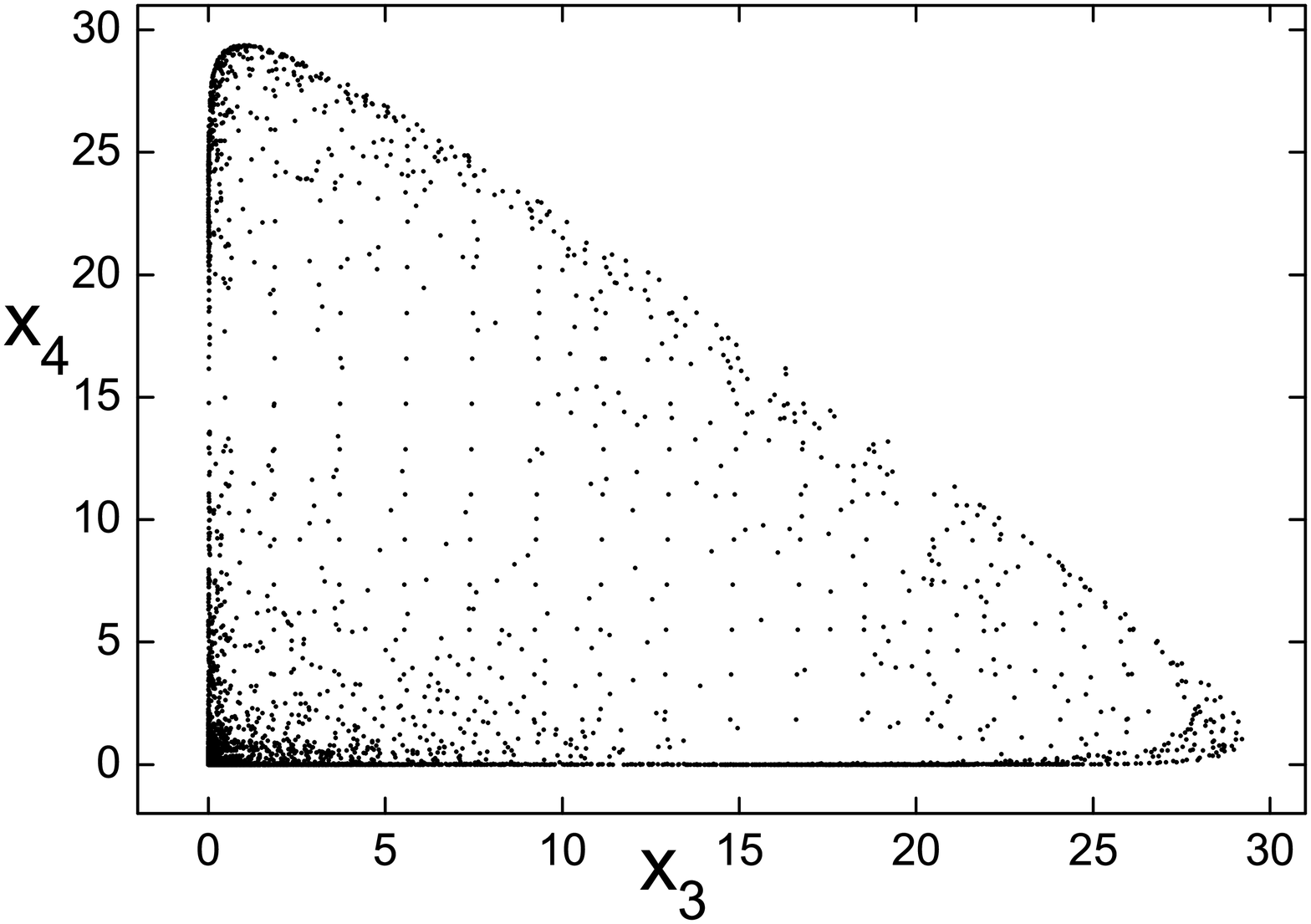}
\caption{The Poincar{\' e} surface of section  {$x_2=1, x_1>1$} for the Lotka--Volterra system 
with $a_i=1$ and $k_i=-1$, $i=1,2,3,4$ for the energies: (a) $E=4.2$, 
(b) $E=6$, (c) $E=8$, (d) $E=29$.    
\label{fig1}    }
\end{figure}

 \begin{figure}
\centering
\includegraphics[scale=0.25]{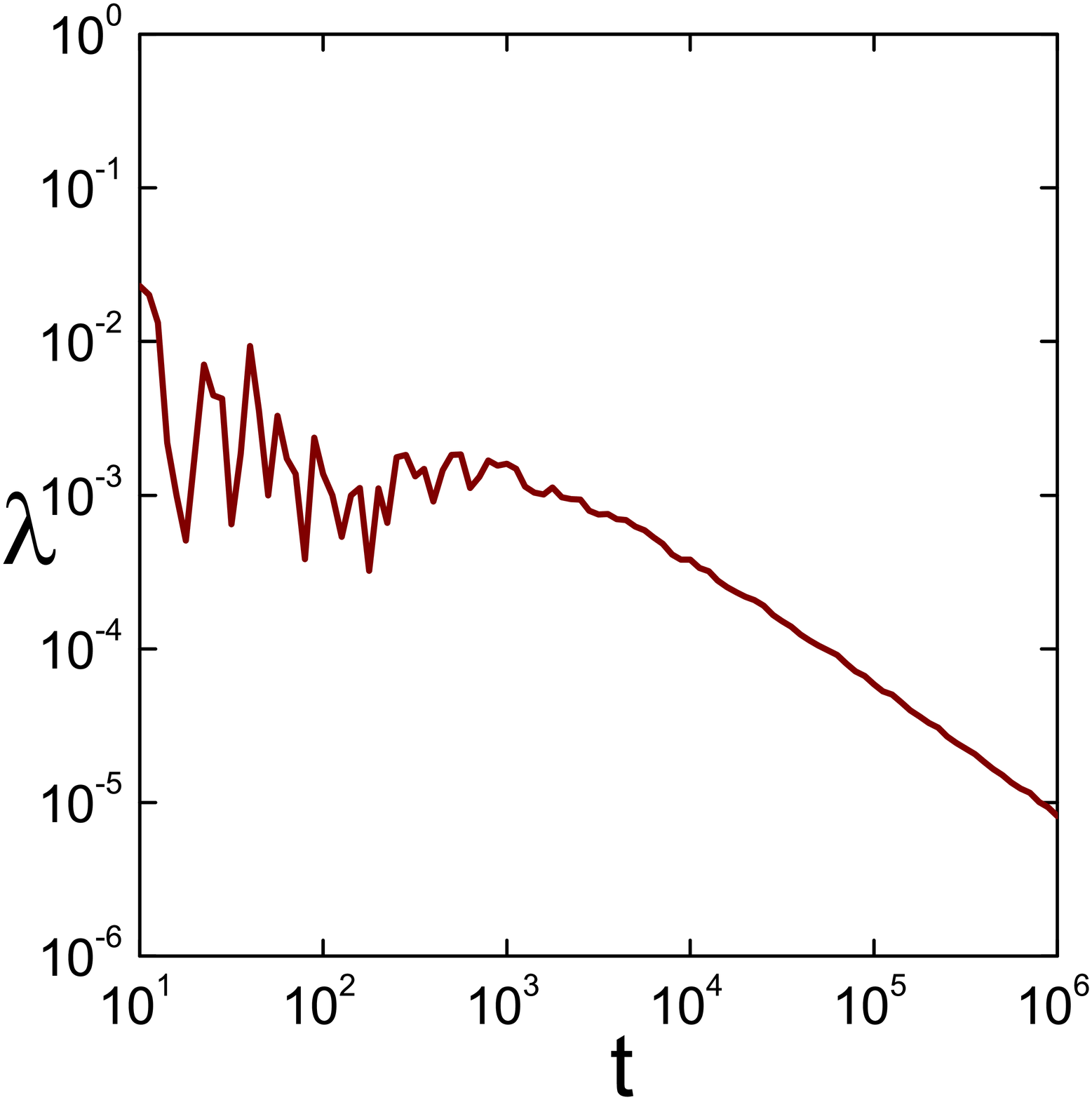}
\includegraphics[scale=0.25]{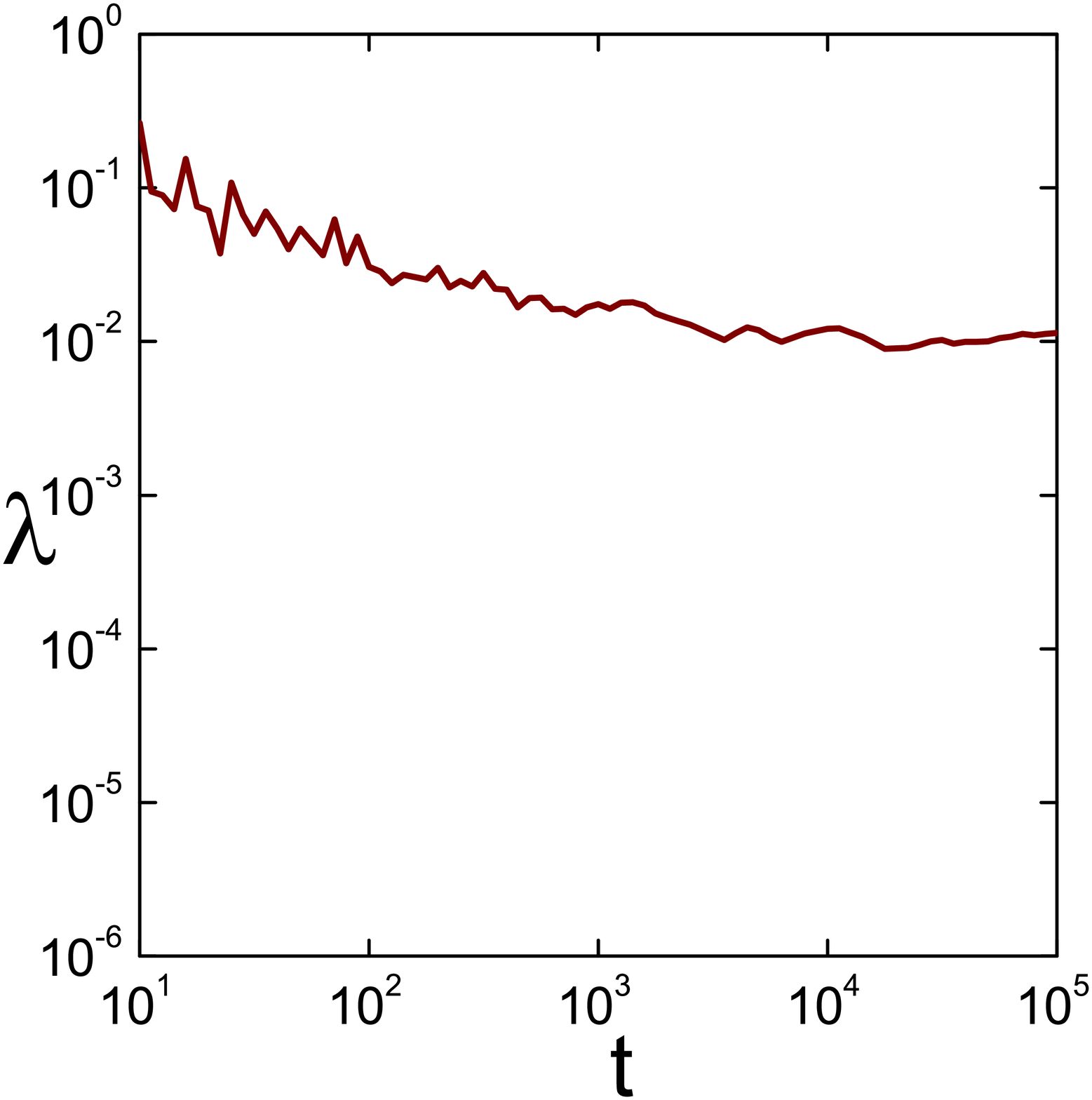}  
\caption{ The largest Lyapunov exponent $\lambda $ for the Lotka--Volterra system 
with $a_i=1$ and $k_i=-1$, $i=1,2,3,4$ for the energies: (a) $E=4.2$ and (b) $E=29$.
\label{fig1b} }
\end{figure}

\begin{figure}
\centering
\includegraphics[scale=0.27]{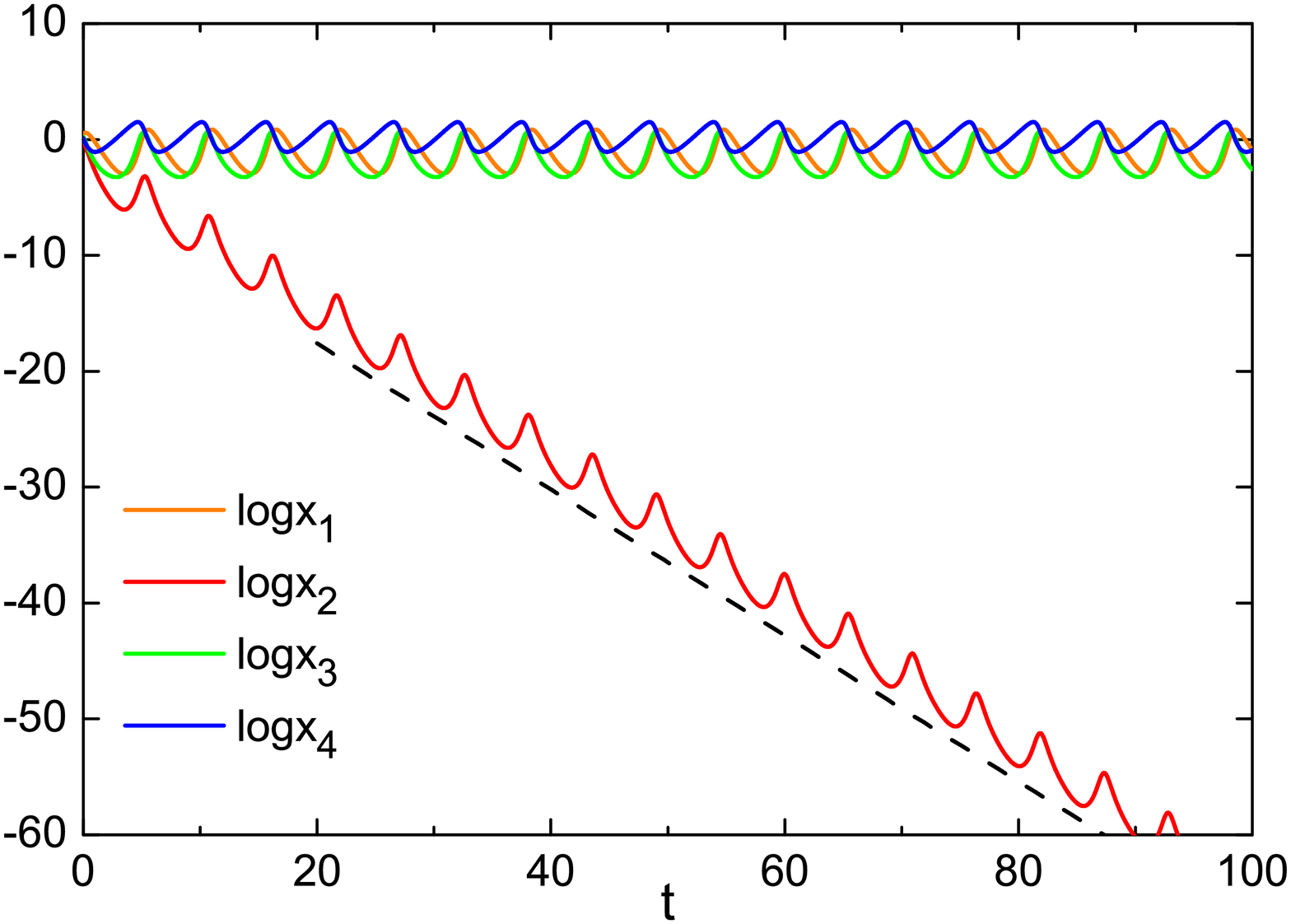}
\includegraphics[scale=0.27]{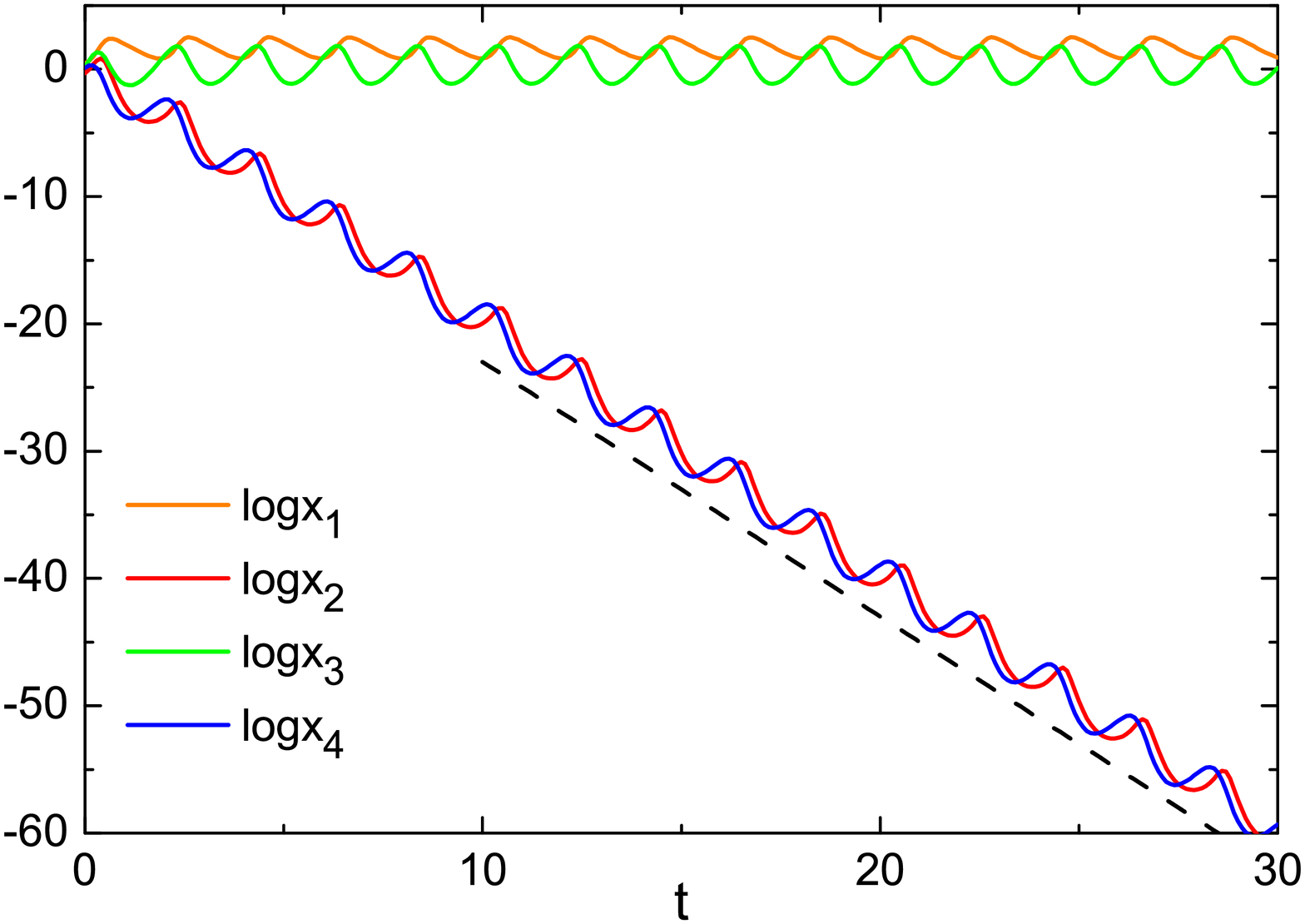} 
\caption{  The evolution in time of the phase space variables for the integrable cases: 
(a) $(k_1,k_2, k_3, k_4) =(0, 0,- 1,-1)$ and (b) $(k_1,k_2, k_3, k_4) =(-2, -2,- 2,2)$.
\label{fig4} }
\end{figure}

\begin{figure}
\centering
\includegraphics[scale=0.25]{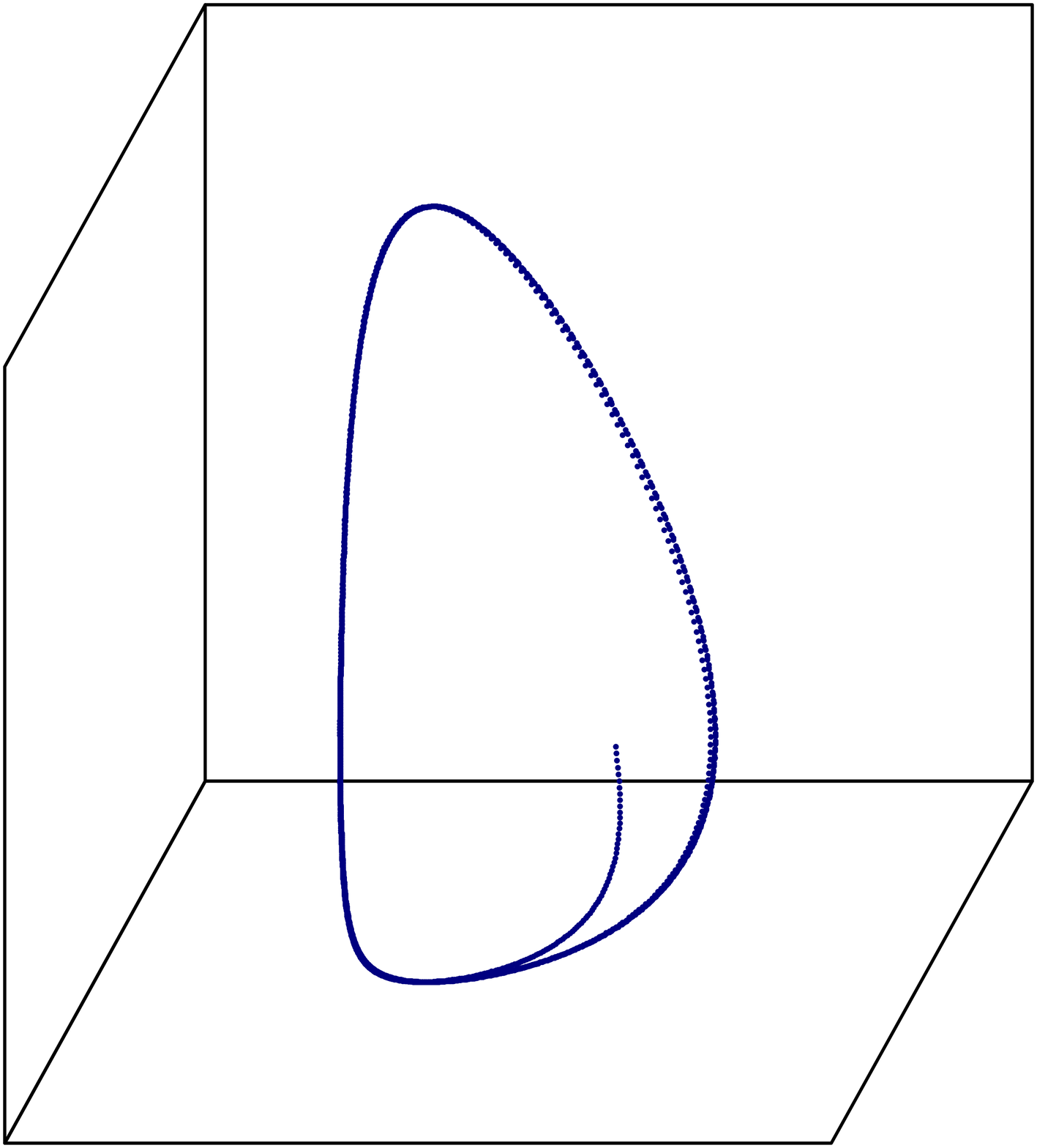}
\includegraphics[scale=0.25]{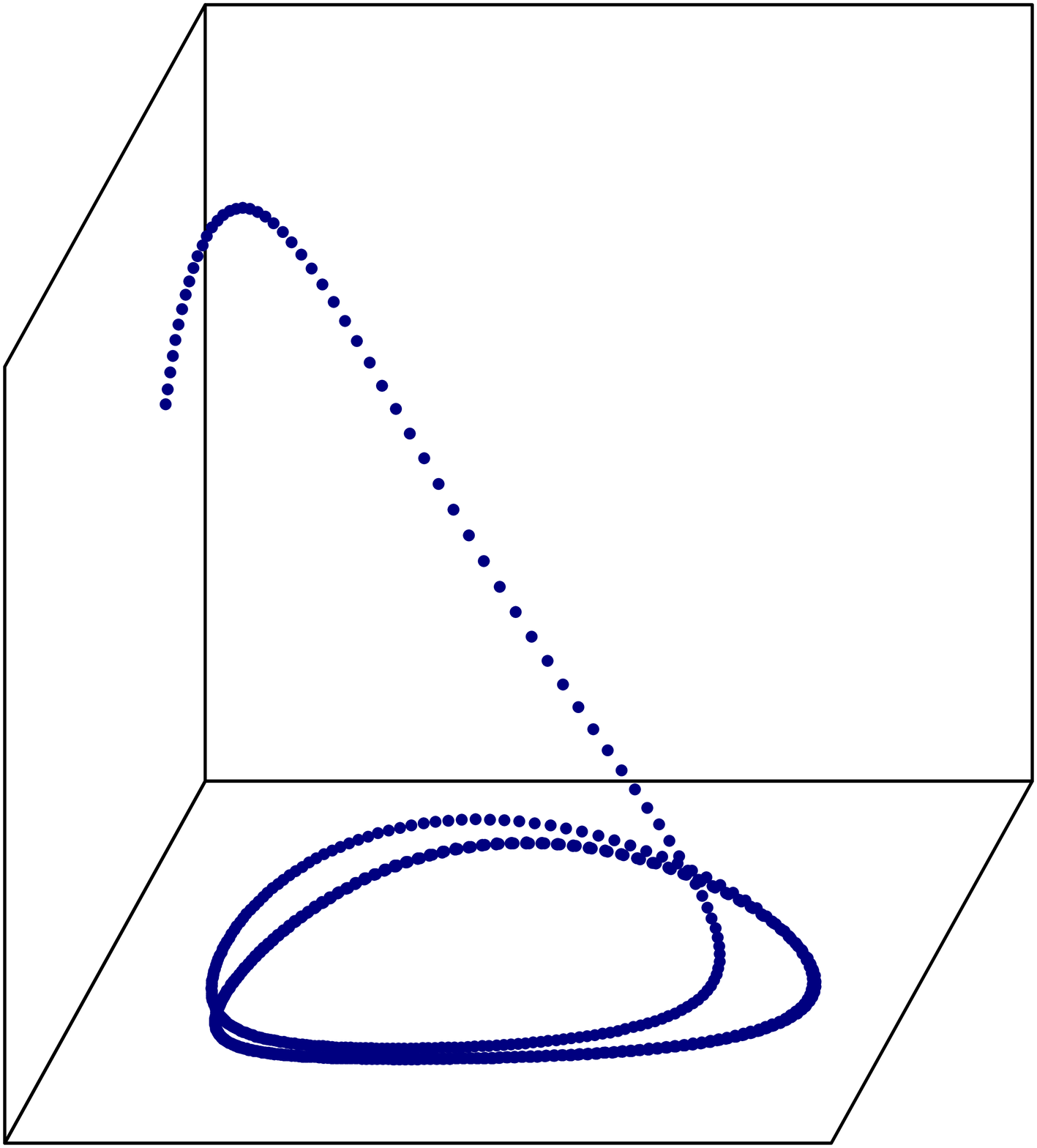} 
\caption{  The trajectories projected on the 3D plane $x_1,x_3,x_4$ plane for the integrable systems:  
(a) $(k_1,k_2, k_3, k_4) =(0, 0,- 1,-1)$ and (b) $(k_1,k_2, k_3, k_4) =(-2, -2,- 2,2)$.
\label{fig4b}   }
\end{figure}

\begin{figure}
\centering
\includegraphics[scale=0.25]{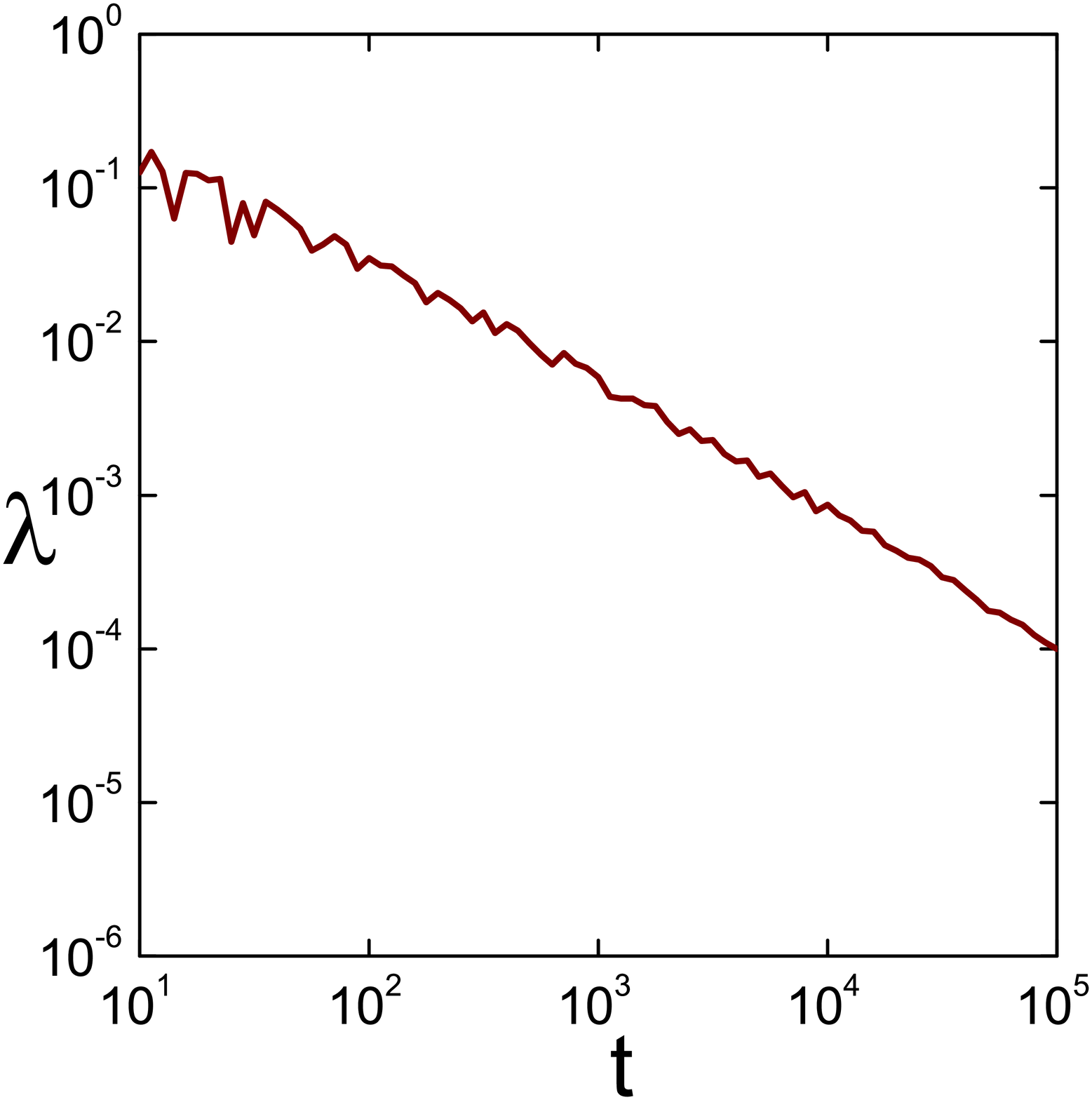}
\includegraphics[scale=0.25]{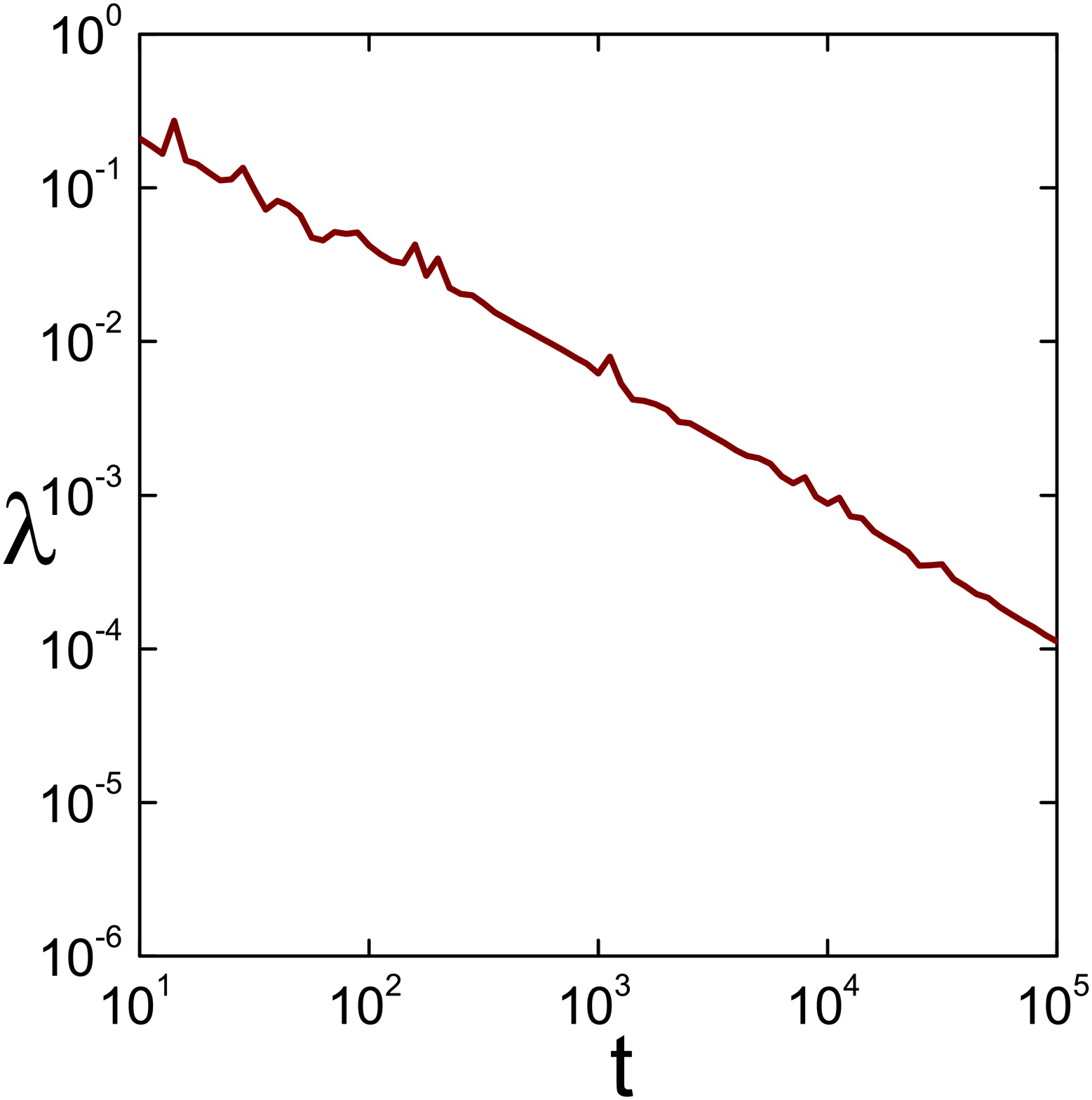}
\caption{   The largest Lyapunov exponent $\lambda $ for the system with $(k_1,k_2, k_3, k_4) =(-2, -2,- 2,2)$ 
at (b) $E=10$ and (c) $E=72$. 
\label{fig5}  }
\end{figure}
We find qualitatively similar behavior to the examples of Fig.\ref{fig2}  
for  $k_1=\dots=k_4=-1$, as Fig.\ref{fig1} indicates. 
In the Poincar{\' e} surface of section $x_2=1$, $x_1>1$ of Fig.\ref{fig1}(a),
 {which corresponds to the energy $E=4.2$,} there is no evidence of chaotic behavior.
We verify this result in Fig.\ref{fig1b}(a)  by computing the largest Lyapunov exponent $\lambda $, 
which approximately decays as $1/t$ for randomly chosen initial conditions. 
Similarly with the well-known H{\' e}non--Heiles model \cite{henon}, chaotic dynamics 
in the Lotka--Volterra system (\ref{eqss}) for $k_i<0$ (or $k_i=-1$)  emerges
for larger values of the energy. In the rest of the panels of Fig.\ref{fig1}, 
where the total energy $E$ is gradually increased, we observe a 
gradual transformation of fixed points and ellipses--like curves, while 
at energies of the order of $E=30$ (Fig.\ref{fig1}(d)) the chaotic motion is not only evident
but also prevails over the ordered motion. The largest Lyapunov exponent at this energy, 
which is plotted in Fig.\ref{fig1b}(b), 
converges to a positive value $\lambda \simeq  0.01$.
 
As we have seen in example 5.3, the only integrable cases for $n=4$, 
$ {\mathbf{a}} =(1,1,1,1)$ predicted by Theorem \ref{lastthm} are for 
$\mathbf{k}=(0,0,k_3,k_4)$,  {$k_3, k_4 \in \mathbb{R}$} or $\mathbf{k}=(k_1,k_2,0,0)$,
 {$k_1, k_2 \in \mathbb{R}$}.
We choose $(k_1,k_2, k_3, k_4) =(0, 0,- 1,-1)$,
for which the quantity $( x_1 +  x_2) x_4 /x_3$ is preserved  {besides} the Hamiltonian. 
Fig.\ref{fig4}(a) displays the evolution of the four variables $\log x_i$ in time for a random choice of 
initial conditions. 
It turns out that $x_2$ decays asymptotically to zero, approximately like $e^{-0.63t}$, while the rest variables 
$x_1, x_3, x_4$ asymptotically approach a periodic orbit, as is illustrated in Fig.\ref{fig4b}(a).
However, a similar behavior appears in other cases, not described as integrable by Theorem \ref{lastthm}.
Such an example is given in Fig.\ref{fig4}(b) and corresponds to $(k_1,k_2, k_3, k_4) =(-2, -2,- 2,2)$.
It turns our that the variables $x_2$ and $x_4$ tend asymptotically to zero as $e^{-2t}$, while
$x_1$ and $x_3$ asymptotically converge to the periodic orbit shown in Fig.\ref{fig4b}(b). 
Furthermore, we carefully examine the largest Lyapunov exponent $\lambda $ in Fig.\ref{fig5} 
for constantly increasing energies and we find that $\lambda \propto 1/t$, even when $E=72$,
which strongly indicates that the system is integrable in this case too.

 Similarly to the case  $(k_1,k_2, k_3, k_4) =(-2, -2,- 2,2)$ we find other cases which display 
integrable behavior, manifested by asymptotically vanishing  Lyapunov exponents.
Few of the cases that we checked are listed in the following table \begin{center}
\begin{tabular}{|c|c|c|c|}
\hline
$k_1$ & $k_2$ & $k_3$  & $k_4$  \\ \hline
 ~1 & ~1 & -1 & -1 \\ 
~1 & -1 & ~1 & -1 \\
 ~1 & -1 & -1 & ~1 \\
-1 & -1 & ~1 & -1 \\
-1 & ~1 & -1 & -1 \\
~1 & -1 & -1 & -1 \\ \hline
\end{tabular}
 \end{center}
 
 Finally, based on our numerical findings 
and observations, we conjecture that chaotic motion for the $n=4$ system (\ref{system2}) 
emerges when $a_i>0$, $k_i<0$ and $a_i<0$, $k_i>0$.

\section{Conclusions}
We presented a new class of Hamiltonian parametric Lotka--Volterra systems with non-zero linear terms and we proved that, for particular choices of parameters, Liouville integrability and superintegrability is established.  
Different choices of parameters when $n=4$, not described by the theory, were studied numerically, 
showing that  both chaotic and new integrable cases appear. Concerning these new cases with integrable behavior,
 {we aim to study them in detail in order to detect additional integrals and complete our investigation by including all the odd dimensional cases too.}
  
In the present work we restricted our analysis to the even-dimensional case; 
however, a similar approach can be considered for odd dimensions. 
 Finally,  we believe that a similar approach can be considered for integrable Lotka--Volterra systems with different community matrices, or integrable deformations of them such as the systems presented in \cite{DEKV,evri,evri2},  by inserting parametric linear terms in the corresponding vector fields.

\section*{Acknowledgements} 
HC is supported by the State Scholarship Foundation (IKY) 
operational Program: `Education and Lifelong Learning--Supporting Postdoctoral Researchers'
 2014-2020, and is co--financed by the European Union and Greek national funds; she is also grateful to SMSAS, Kent for hosting her as a visitor. 
ANWH is supported by Fellowship EP/M004333/1 from the Engineering $\&$ Physical Sciences 
Research Council, UK, and is grateful to the School of Mathematics $\&$ Statistics, UNSW for hosting him as a 
Visiting Professorial Fellow with funding from  
the Distinguished Researcher Visitor scheme; he also thanks Prof. Wolfgang Schief for additional financial support in 2019.  
TEK would like to thank Prof. Reinout Quispel, Dr  Peter Van Der Kamp and Dr Charalambos Evripidou for their hospitality at La Trobe University, and for their useful 
comments on this topic.

\appendix
\section{ {Comments and examples on the odd dimensional cases}} \label{sectionodd}
As it is stated in Section \ref{secHamForm}, in the odd dimensional cases the described Hamiltonian formalism, i.e. the log-canonical Poisson structure \eqref{poisson} along with the Hamiltonian $H(\mathbf{x})=\sum_{i=1}^n (a_i x_i+k_i\log x_i),$ 
is not sufficient to include all the cases of  vector fields \eqref{system} for arbitrary $r_i$, since matrix \eqref{prtder} is not invertible. 
 Therefore, in this setting we can only restrict to the cases with $\mathbf{r}=P\mathbf{k}$, that is systems of the form \eqref{system2}. 
 For $n=3$, the integrability of  \eqref{system2} follows directly from its Hamiltonian formalism and the existence of the Casimir function $\frac{x_1 x_3}{ x_2}$. 
 More interesting integrable cases emerge for odd $n>3$, by considering the corresponding integrals 
of the $\mathbf{k}=0$ case as they appear in \cite{KQV} and the corresponding permutation symmetry of the system. We will illustrate this in the following example for $n=5$. 

Let us consider the system
\begin{equation} \label{systemn5}
  \dot{x_i}=x_i \left(\sum_{j=1}^5 P_{ij}( a_j x_j+k_j) \right)\;,  \ \ i=1,\dots,5,
\end{equation}
with parameters $\mathbf{a}=(a_1,\dots,a_5),\mathbf{k}=(k_1,\dots,k_5) \in \mathbb{R}^5$. According to \cite{KQV}, for $\mathbf{k=0}$ this system 
admits the first integral
\begin{equation*}
F=\frac{x_5}{x_4}(a_1 x_1+a_2 x_2+a_3 x_3).
\end{equation*}
We compute its Poisson bracket with the Hamiltonian $H=\sum_{i=1}^n (a_i x_i+k_i\log x_i)$  of  
\eqref{systemn5} to get
\begin{equation*}
\{F,H\}=\frac{x_5}{x_4}\left(a_1(k_2+k_3)x_1+a_2(k_3-k_1)x_2- a_3(k_1+k_2) \right).
\end{equation*}
Hence, $F$ is a first integral of \eqref{systemn5} if and only if 
\begin{equation} \label{Fn5}
a_1(k_2+k_3)=a_2(k_3-k_1)=a_3(k_1+k_2)=0.
\end{equation} 
If the parameters $\mathbf{a},\mathbf{k}$ satisfy \eqref{Fn5}, then the integral $F$ in addition to the Casimir function $C=\frac{x_1 x_3 x_5}{ x_2 x_4}$ ensures 
the  
complete integrability of the system. 
Furthermore, the invariance of \eqref{systemn5} under the transformation $x_i \mapsto x_{6-i}$, $a_i \mapsto -a_{6-i}$, $k_i \mapsto -k_{6-i}$, 
implies that 
\begin{equation*}
\tilde{F}=\frac{x_1}{x_2}(a_5 x_5+a_4 x_4+a_3 x_3).
\end{equation*}
is a first integral of \eqref{systemn5} if and only if 
\begin{equation} \label{tildeFn5}
a_5(k_4+k_3)=a_4(k_3-k_5)=a_3(k_5+k_4)=0. 
\end{equation}
So we conclude that system \eqref{systemn5} is integrable if the parameters $\mathbf{a},\mathbf{k}$ satisfy \eqref{Fn5} or \eqref{tildeFn5}. 

For example, in the case of $\mathbf{a} \neq 0$, system \eqref{systemn5} is integrable if $k_3=-k_2=k_1$ or $k_3=-k_4=k_5$, while the case of 
$k_5=-k_4=k_3=-k_2=k_1$ which leads to superintegrability is equivalent to the $\mathbf{k}=0$ case.

\end{document}